\documentclass[journal]{IEEEtran}
\usepackage{amsmath,amsfonts, amsmath, amsthm, amssymb}
\usepackage{algorithmic}
\usepackage{algorithm}
\usepackage{array}
\usepackage{subcaption}
\usepackage{textcomp}
\usepackage{stfloats}
\usepackage{url}
\usepackage{verbatim}
\usepackage{graphicx}
\usepackage{cite}
\usepackage{epsfig}
\usepackage{svg}
\usepackage{physics}
\usepackage{hyperref}
\usepackage{xcolor}
\usepackage{bbm}
\usepackage{multirow}
\usepackage{makecell}

\DeclareMathOperator*{\argmin}{arg\,min}

\newcommand{\diag}{{\mathrm{diag}}}

\newtheorem{definition}{Definition}

\newtheorem{theorem}{Theorem}
\newtheorem{corollary}[theorem]{Corollary}
\newtheorem{lemma}[theorem]{Lemma}
\newtheorem{proposition}[theorem]{Proposition}

\newtheorem{remark}{Remark}

\newcommand{\supp}{\mathrm{supp}}
\newcommand{\cov}{\mathrm{Cov}}
\newcommand{\Var}{\mathrm{Var}}

\usepackage{tikz, nicematrix, mathtools, calc}
\usepackage{arydshln}
\usepackage{standalone}
\usetikzlibrary{arrows.meta, positioning, calc, fit, shapes.geometric, tikzmark}
\tikzset{
    m/.style = {anchor=base, text width=1.6em, align=center, inner sep=0pt},
    arr/.style = {->,>=Stealth,shorten >=1pt,shorten <=1pt, line width = 0.5pt},
    box/.style = {draw, line width = 0.5pt, minimum width=1.4cm, minimum height=0.8cm, align=center},
    lbox/.style = {draw, line width = 0.5pt, minimum width=1.8cm, minimum height=0.8cm, align=center},
    xs/.style = {circle,fill=black!80,inner sep=2.5pt},
    xn/.style = {diamond, draw = black!80, line width = 1pt, inner sep = 2pt, rounded corners = 1pt},
    legend/.style = {draw, rounded corners=2pt, fill=white, inner sep=3pt, align=left, font=\small, line width=.6pt}
}
\NiceMatrixOptions{code-for-first-row = \color{red}, code-for-first-col = \color{blue}}

\begin{document}

\title{Fundamental Limit of Discrete Distribution Estimation under Utility-Optimized Local Differential Privacy

\author{
    Sun-Moon~Yoon,~\IEEEmembership{Graduate Student Member,~IEEE,}
    Hyun-Young~Park,~\IEEEmembership{Graduate Student Member,~IEEE,}
    Seung-Hyun~Nam,~\IEEEmembership{Member,~IEEE,}
    and Si-Hyeon~Lee,~\IEEEmembership{Senior Member,~IEEE}
}% <-this % stops a space
\thanks{

}
\thanks{S.-M. Yoon, H.-Y. Park, and S.-H. Lee are with the School of Electrical Engineering, Korea Advanced Institute of Science and Technology (KAIST), Daejeon, South Korea (e-mail: ysm4078@kaist.ac.kr, phy811@kaist.ac.kr, \mbox{sihyeon@kaist.ac.kr}). S.-H. Nam is with the Information \& Electronics Research Institute, KAIST, Daejeon, South Korea (e-mail: shnam@kaist.ac.kr).  S.-M. Yoon and H.-Y. Park contributed equally to this work. (Corresponding author: Si-Hyeon Lee)
}
\thanks{The appendix of this paper is provided in the supplementary material.}
}

\maketitle
\begin{abstract}
We study the problem of discrete distribution estimation under utility-optimized local differential privacy (ULDP), which enforces local differential privacy (LDP) on sensitive data while allowing more accurate inference on non-sensitive data. In this setting, we completely characterize the fundamental privacy–utility trade-off. The converse proof builds on several key ideas, including a generalized uniform asymptotic Cramér–Rao lower bound, a reduction showing that it suffices to consider a newly defined class of extremal ULDP mechanisms, and a novel distribution decomposition technique tailored to ULDP constraints. For the achievability, we propose a class of utility-optimized block design (uBD) schemes, obtained as nontrivial modifications of the block design mechanism known to be optimal under standard LDP constraints, while incorporating the distribution decomposition idea used in the converse proof and a score-based linear estimator. These results provide a tight characterization of the estimation accuracy achievable under ULDP and reveal new insights into the structure of optimal mechanisms for privacy-preserving statistical inference.
\end{abstract}


\section{Introduction}
Local differential privacy (LDP) \cite{kasiviswanathan2011can, duchi2013local, dwork2014algorithmic} has already been deployed  at scale in industry \cite{ding2017collecting, apple2017privacy_scale, apple2025aggregate_trends, erlingsson2014rappor, bittau2017prochlo}.
Under LDP, each individual perturbs their data locally before transmission so that the reported value reveals provably little about the true record, thereby removing the need for a trusted curator \cite{duchi2013local, kasiviswanathan2011can}.
Formally, LDP bounds how much the distribution of a user’s reported value can differ as their true data vary, thereby embedding privacy directly at the source. 
This user-level privacy guarantee—quantified through information-theoretic bounds \cite[Thm. 14]{issa20_operational}—has propelled LDP into a wide spectrum of statistical inference and machine-learning tasks involving sensitive data \cite{warner1965randomized, erlingsson2014rappor, ye2018optimal, acharya2019hadamard, feldman2022private, wang2017locally}.

However, LDP offers blanket protection: for any pair of input values, the distribution of reported data must remain nearly indistinguishable \cite{issa20_operational}.
Although such protection is reassuring, it can be inefficient in scenarios where only certain inputs are truly sensitive. 
Consider the survey, ``Have you violated a policy?'', with $X\in \{\text{``Yes''},\text{``No''} \}$ as the truthful answer, adapted from the illustrative example in \cite{murakami2019utility}.
Under Warner’s randomized response (RR) \cite{warner1965randomized} shown in Fig.~\ref{fig:rr} (a), respondents flip a biased coin and potentially report the opposite answer.
Consequently, the reported value $Y$ conveys only limited information about $X$, restricting the analyst’s ability to infer the true outcome.
Notably, however, only the affirmative response (``Yes") is genuinely sensitive, while the negative response (``No") is not.
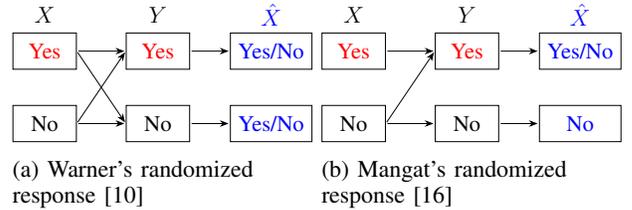
\begin{figure}[t!]
    \centering
    \begin{subfigure}{.45\linewidth}
        \makebox[\linewidth][c]{\begin{tikzpicture}[x=1cm,y=1cm]

\node[font=\Large] at (0,1.6) {$X$};
\node[font=\Large] at (2.5,1.6) {$Y$};
\node[font=\Large, text=blue] at (5,1.6) {$\hat X$};

\node[box]  (Xyes) at (0, 0.8) {\textcolor{red}{\Large Yes}};
\node[box]  (Xno)  at (0,-0.8) {\Large No};

\node[box] (Yyes) at (2.5, 0.8) {\textcolor{red}{\Large Yes}};
\node[box] (Yno)  at (2.5,-0.8) {\Large No};

\node[lbox] (Xhat1) at (5, 0.8) {\textcolor{blue}{\Large Yes/No}};
\node[lbox] (Xhat2) at (5,-0.8) {\textcolor{blue}{\Large Yes/No}};

\draw[arr] (Xyes.east) -- (Yyes.west);
\draw[arr] (Xno.east)  -- (Yno.west);
\draw[arr] (Xyes.east) -- (Yno.west);
\draw[arr] (Xno.east)  -- (Yyes.west);

\draw[arr] (Yyes.east) -- (Xhat1.west);
\draw[arr] (Yno.east)  -- (Xhat2.west);

\end{tikzpicture}}
        \caption{Warner's randomized\\ response \cite{warner1965randomized}}
    \end{subfigure}
    \begin{subfigure}{.45\linewidth}
        \makebox[\linewidth][c]{\begin{tikzpicture}[x=1cm,y=1cm]

\node[font=\Large] at (0,1.6) {$X$};
\node[font=\Large] at (2.5,1.6) {$Y$};
\node[font=\Large, text=blue] at (5,1.6) {$\hat X$};

\node[box]  (Xyes) at (0, 0.8) {\textcolor{red}{\Large Yes}};
\node[box]  (Xno)  at (0,-0.8) {\Large No};

\node[box] (Yyes) at (2.5, 0.8) {\textcolor{red}{\Large Yes}};
\node[box] (Yno)  at (2.5,-0.8) {\Large No};

\node[lbox] (Xhat1) at (5, 0.8) {\textcolor{blue}{\Large Yes/No}};
\node[lbox] (Xhat2) at (5,-0.8) {\textcolor{blue}{\Large No}};

\draw[arr] (Xyes.east) -- (Yyes.west);
\draw[arr] (Xno.east)  -- (Yno.west);
\draw[arr] (Xno.east)  -- (Yyes.west);

\draw[arr] (Yyes.east) -- (Xhat1.west);
\draw[arr] (Yno.east)  -- (Xhat2.west);

\end{tikzpicture}}
        \caption{Mangat's randomized\\ response \cite{mangat1990alternative}}
    \end{subfigure}
    \caption{Comparison between two binary randomized response models, where $\hat{X}$ denotes a guess about $X$ based on $Y$.}
    \label{fig:rr}
\end{figure}

Mangat’s randomized response \cite{mangat1990alternative} in Fig.~\ref{fig:rr} (b), exploits this asymmetry: those who answered ``Yes'' always report truthfully, while those who answered ``No'' randomize their answer.
This ensures privacy when $Y = \text{``Yes''}$, while enabling reliable inference when $Y = \text{``No''}$.
Utility-optimized local differential privacy (ULDP) \cite{murakami2019utility} formalizes this idea.
A ULDP mechanism partitions outputs into protected symbols—subject to the full LDP guarantee—and invertible symbols that faithfully reveal non-sensitive inputs.
Thus, ULDP preserves LDP-level privacy exactly where needed, while permitting enhanced inference elsewhere.
This approach has motivated extensions such as high–low LDP \cite{acharya2020context}, $(\epsilon,\, \delta)$-ULDP \cite{zhang2023frequency,zhang2024mean}, input-discriminative LDP \cite{gu2020providing, murakami2023automatic}, and utility-optimized key-value data LDP \cite{wang2025ukvldp}. 
Furthermore, ULDP has already seen adoption in domains including clinical analytics \cite{bonomi2023enabling}, federated learning \cite{islam2022intelligent}, and natural language processing \cite{yue2021differential}.

For the representative statistical inference task of discrete distribution estimation, which is the focus of this paper, several ULDP schemes have been proposed \cite{murakami2019utility, liu2024multidimensional, acharya2020context, he2025addressing}. 
As will be elaborated in Section~\ref{sec:preliminaries}, prior ULDP schemes have been analyzed only with respect to order-optimality in specific regimes \cite{murakami2019utility, acharya2020context, he2025addressing}, while the exact leading constant that characterizes the fundamental limit of the PUT under ULDP has not been identified. 
This unresolved gap motivates our work, which establishes the exact attainable performance.

\subsection{Our Contribution}
Our main contribution is to establish the exact optimal PUT—including the leading constant—for discrete distribution estimation under ULDP. 
We characterize the optimal PUT in terms of a finite-dimensional concave–convex saddle-point optimization problem that can be solved using convex optimization solvers and give a closed-form solution to this optimization problem for some regime. 

For converse, we develop two general \emph{analytical frameworks} which extend previous works: (i) a \emph{uniform asymptotic Cram\'er-Rao lower bound}, a versatile concept applicable beyond ULDP to general private estimation problems, and (ii) a reduction showing that optimal ULDP mechanisms can be restricted to a newly defined class of \emph{extremal ULDP mechanisms}, analogous to extremal LDP mechanisms \cite{kairouz2016extremal}, thereby confining the minimax analysis to this smaller class.
After that, we complete the converse proof by exploiting a novel way to \emph{decompose the data distribution} into several parts based on the structure of the ULDP constraint.

For achievability, we propose a class of \textbf{utility-optimized block design (uBD) schemes}. The development of uBD schemes builds on several key ideas: (i) nontrivially modifying \emph{block design schemes} \cite{park2023block} for discrete distribution estimation under LDP, (ii) exploiting the idea of \emph{score-based linear estimator} in \cite{namOptimalPrivateDiscrete2024}, and 
(iii) \emph{decomposing the data distribution} as in the converse proof.
Notably, the class of uBD schemes subsumes one of the existing schemes, uRR \cite{murakami2019utility}, and our closed-form characterization implies that uRR is optimal in some regime.
To the best of our knowledge, this is the first work to prove the optimality of uRR in a certain regime. Furthermore, we show that there exists a regime in which previously proposed ULDP schemes cannot achieve the optimal PUT, thereby establishing the strict superiority of our uBD scheme.

\subsection{Paper Outline}
The rest of this paper is organized as follows. 
Section~\ref{sec:prob_form} formalizes the discrete distribution estimation problem under the ULDP constraint.
Section~\ref{sec:preliminaries} provides the necessary background and a review of the related works.
Our main result—the characterization of the optimal PUT—is stated in Section~\ref{sec:main}.
To establish the optimal PUT, Section~\ref{sec:converse} presents the converse proof, and Section~\ref{sec:achievabilityuBD} establishes achievability through the proposed uBD scheme.
Section~\ref{sec:experiment} provides the numerical evaluations of the proposed scheme and baselines on a real-world dataset, demonstrating that our scheme outperforms existing alternatives.
Finally, Section~\ref{sec:concl} concludes the paper with some discussion on future works. 

\subsection{Notations and Terminologies}
For natural numbers $a \leq b$, we denote $[a:b]:=[a,b]\cap \mathbb{N}$ and $[a]:= [1:a]$.
For a finite set $\mathcal{X}$, $\mathcal{P}(\mathcal{X})$ denotes the set of all probability mass functions on $\mathcal{X}$. 
For $x \in \mathcal{X}$ and $P \in \mathcal{P}(\mathcal{X})$, we write either $P(x)$ or $P_x$ for $P(\{x\})$.
When $\mathcal{X}=[w]$ for $w \in \mathbb{N}_{\geq 2}$, we identify $\mathcal{P}(\mathcal{X})$ with the probability simplex,
\begin{equation}
    \!\! \Delta_{w} := \left\{P \in \mathbb{R}^{w} : \sum\limits_{x=1}^w P_x = 1 \text{ and } \forall x \in [w], P_x \ge 0  \right\}.
\end{equation}
The relative interior of $\Delta_{w}$ is
\begin{align}
    \Delta_w^\circ := \{P \in \Delta_w : P_x > 0, \quad \forall x \in [w]\}, \label{eq:relIntDef}
\end{align}
and its direction space is
\begin{align}
    \vec{\Delta}_w := \left\{h \in \mathbb{R}^{w} : \sum\limits_{x=1}^w h_x = 0 \right\}. \label{eq:directSpaceDef}
\end{align}
For $x \in [w]$, we write $\delta^{(x;w)} \in \Delta_w$ for the point mass at $x$, i.e., $\delta^{(x;w)}_{x'} = \mathbbm{1}(x=x')$.

For a conditional distribution $Q:\mathcal{X} \rightarrow \mathcal{P}(\mathcal{Y})$ between finite sets $\mathcal{X},\mathcal{Y}$, we denote $\supp(Q):=\{y \in \mathcal{Y}:Q(y|x)>0 \text{ for some }x \in \mathcal{X}\}$.
For simplicity, we often regard $Q$ as the row stochastic matrix $(Q_{x,y})_{(x,y) \in \mathcal{X}\times \mathcal{Y}}$ with entries $Q_{x,y} = Q(y|x)$.
Given an input distribution, $P\in\mathcal{P}(\mathcal{X})$, we write $Q_P \in \mathcal{P}(\mathcal{Y})$ for the marginal distribution of the output, that is, $Q_P(y) = \sum_{x\in\mathcal{X}}Q(y|x)P_x$.

We write $\mathbf{0}_{m}, \mathbf{1}_{m}\in\mathbb{R}^m$ as the all-zeros and all-ones vectors (and $\mathbf{0}_{m\times n},\mathbf{1}_{m\times n}$ for their matrix analogues), and $\mathbf{I}_m$ as the $m\times m$ identity matrix. 
For $\gamma \in \mathbb{R}^m$, we write $\diag(\gamma)$ as the $m \times m$ diagonal matrix whose $i$-th diagonal element is $\gamma_i$.
For $A\in\mathbb{R}^{n\times m}$, $S_1 \subset [n]$, and $S_2 \subset [m]$, we write $A_{S_1,S_2}$ for the submatrix of $A$ with rows indexed by $S_1$ and columns indexed by $S_2$. 

\newcommand{\trInv}{{\tr_{-1}}}
For a positive definite matrix $A$, we define $\trInv(A):=\tr(A^{-1})$, and for a non-invertible positive semidefinite matrix $A$, we define $\trInv(A):=\infty$.

\section{Problem Formulation}\label{sec:prob_form}
{
\begin{figure}%[h!]
    \centering
    \begin{subfigure}{.75\linewidth}
        \makebox[\linewidth][c]{\begin{tikzpicture}[x=1cm,y=1cm]
\node[m] (X1) at (0,3) {$X_1$};
\node[m] (X2) at (0,2) {$X_2$};
\node[m] (Xn) at (0,0) {$X_n$};
\node[m, font = \Large] (dot) at (0,1) {$\vdots$};
\coordinate (c1n) at ($(X1)!0.5!(Xn)$);
\node[m, left = 7mm of c1n] (P) {$P$};
\draw[line width = 0.5pt] (P.east) -- ($(P.east)+(2mm,0)$) -- ($(P.east)+(2mm,0)+(0,1.5)$) -- (X1.west);
\draw[line width = 0.5pt] ($(P.east)+(2mm,0)+(0,0.5)$) -- (X2.west);
\draw[line width = 0.5pt] ($(P.east)+(2mm,0)$) -- ($(P.east)+(2mm,0)+(0,-1.5)$) -- (Xn.west);

\node[draw, rectangle, right=5mm of X1, minimum width=1em] (Q1) {$Q$};
\node[draw, rectangle, right=5mm of X2, minimum width=1em] (Q2) {$Q$};
\node[draw, rectangle, right=5mm of dot, minimum width=1em, draw opacity=0, text opacity=0, fill opacity=0] (Qghost) {$Q$};
\node[draw, rectangle, right=5mm of Xn, minimum width=1em] (Qn) {$Q$};

\draw[arr] (X1.east) -- (Q1.west);
\draw[arr] (X2.east) -- (Q2.west);
\draw[arr] (Xn.east) -- (Qn.west);

\node[m, right=5mm of Q1] (Y1) {$Y_1$};
\node[m, right=5mm of Q2] (Y2) {$Y_2$};
\node[m, right=5mm of Qn] (Yn) {$Y_n$};
\node[m, font=\Large, right = 5mm of Qghost] {$\vdots$};

\draw[arr] (Q1.east) -- (Y1.west);
\draw[arr] (Q2.east) -- (Y2.west);
\draw[arr] (Qn.east) -- (Yn.west);

\coordinate (SrvX) at ([xshift=8mm]Y2.east);

\node[draw, inner sep=15pt, fit={($(SrvX |- Y1)+(0,-0.2)$) ($(SrvX |- Yn)+(0,0.2)$)}] (Server) {};
\node at (Server) {server};

\draw[line width=0.5pt] (Y1.east) -- (Server.west |- Y1.east);
\draw[line width=0.5pt] (Y2.east) -- (Server.west |- Y2.east);
\draw[line width=0.5pt] (Yn.east) -- (Server.west |- Yn.east);

\node[m,anchor = center] at ($(Server.east)+(6mm,0)$) (h) {$\hat{P}$};
\draw[line width = 0.5pt] (Server.east) -- (h.west);
\end{tikzpicture}} \caption{A discrete distribution estimation system under the ULDP constraint.}
    \end{subfigure} 
    \begin{subfigure}{.45\linewidth}
    \makebox[\linewidth][c]{\begin{tikzpicture}[x=1cm,y=1cm,>=Stealth]
\coordinate (A) at (0,0);    
\coordinate (B) at (4,2);     
\draw[dashed, line width = 0.5pt] (1, 0.2) -- (1, 0) -- (A) -- (0,2) -- (B) -- (4,1.2) -- (3.3, 1.2)-- cycle;

\node[above] at ($(A|-B)!0.5!(B)$) {$\epsilon$-LDP};

\node (XS) at ($(A)+(0.4,1.6)$) {$\mathcal{X}_{\mathrm{S}}$};
\node (XN) at ($(A)+(0.4,0.4)$) {$\mathcal{X}_{\mathrm{N}}$};

\node (YP) at ($(B)+(-0.4,-0.4)$) {$\mathcal{Y}_{\mathrm{P}}$};
\node (YI) at ($(A)+(3.6,0.4)$)   {$\mathcal{Y}_{\mathrm{I}}$};

\draw[arr]            (XS.east) -- (YP.west);     
\draw[arr]            (XN.east) -- (YI.west);      
\draw[arr]     (XN.east) -- ($(YP.west)-(0,1mm)$);      

\end{tikzpicture}}
    \caption{Illustration of a ULDP mechanism.}
    \end{subfigure}

    \caption{Overview of system model.}
    \label{fig:systemmodel}
\end{figure}
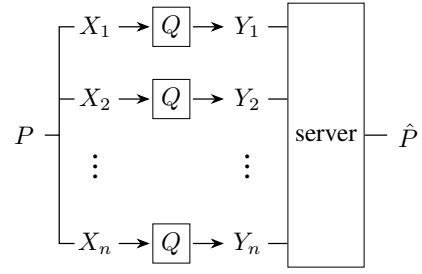
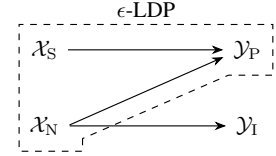
}
We consider the private discrete distribution estimation system with $n\in\mathbb{N}$ clients and a single server, as depicted in Fig.~\ref{fig:systemmodel}. 
For each $i\in[n]$, the $i$-th client has its own categorical data $X_i \in \mathcal{X} = [w]$, $w \in \mathbb{N}_{\geq 2}$, generated in an i.i.d. manner from an unknown distribution $P \in \mathcal{P}(\mathcal{X})$. 
To mitigate privacy leakage, the $i$-th client perturbs its data $X_i$ into $Y_i$ through a conditional distribution $Q:\mathcal{X} \rightarrow \mathcal{P}(\mathcal{Y})$ for some finite set $\mathcal{Y}$, called a \textbf{privacy mechanism}, and transmits $Y_i$ to the server.
After collecting $Y^n = (Y_1,\dots,Y_n)$, the server computes the estimate $\hat{P}_n:\mathcal{Y}^n \rightarrow \mathbb{R}^w$ of $P$.
Let $\hat{P}$ denote a sequence of estimators $(\hat{P}_n)_{n=1}^\infty$ and we call the pair $(Q,\hat{P})$ a \textbf{private estimation scheme}.
The \textbf{estimation error} of a private estimation scheme $(Q,\hat{P})$ with $n$ clients is measured by the worst-case mean squared error (MSE),
\begin{equation}
    R_{n}(Q,\hat{P}) := \sup\limits_{P \in \Delta_{w}}R_{n}(Q,\hat{P};P),
\end{equation}
where
\begin{align}
    R_{n}(Q,\hat{P};P) = {\mathbb{E}}_{Y^n \sim Q_P^n}\left[ \| \hat{P}_n(Y^n)-P \|_2^2 \right].
\end{align}

In this paper, we consider the \emph{utility-optimized local differential privacy} (ULDP) constraint assuming a partially sensitive domain: sensitive categories are privatized under an $\epsilon$-LDP whereas non-sensitive categories may be disclosed via invertible outputs to enhance utility. 
Specifically, we suppose that $\mathcal{X}$ is partitioned into two sets: the set of sensitive elements $\mathcal{X}_{\mathrm{S}} \subset \mathcal{X}$ and the set of non-sensitive elements $\mathcal{X}_{\mathrm{N}} = \mathcal{X}\backslash \mathcal{X}_{\mathrm{S}}$.
This partition is common knowledge to all clients and the server.
The ULDP is defined as follows.

\begin{definition}[ULDP \cite{murakami2019utility}]\label{def:ULDP}
For $\mathcal{X}_{\mathrm{S}} \subset \mathcal{X}$ and $\epsilon >0$, a privacy mechanism $Q:\mathcal{X}\rightarrow\mathcal{P}(\mathcal{Y})$ is called an $(\mathcal{X}_{\mathrm{S}},\epsilon)$-\textbf{ULDP mechanism} if there exists a partition $\{\mathcal{Y}_{\mathrm{P}},\mathcal{Y}_{\mathrm{I}}\}$ of $\mathcal{Y}$ (called the sets of \textbf{protected data} and \textbf{invertible data}, respectively), such that
\begin{enumerate}
    \item  For all $x, x' \in \mathcal{X}$ and $y \in \mathcal{Y}_{\mathrm{P}}$,
\begin{equation}\label{eq:ULDP_sensitive}
    Q(y|x) \le e^{\epsilon}Q(y|x').
\end{equation}
    \item For all $y \in \mathcal{Y}_{\mathrm{I}}$, there exists $x \in \mathcal{X}_{\mathrm{N}}$ such that
\begin{equation}\label{eq:ULDP_non_sensitive}
    Q(y|x) > 0,\; Q(y|x') = 0 \text{ for all } x' \in \mathcal{X}, x' \neq x.
\end{equation}
\end{enumerate}
\end{definition}
Without loss of generality, we assume that the set of sensitive elements is given as $\mathcal{X}_{\mathrm{S}} = [v]$, $v \leq w$, and we simply call an $(\mathcal{X}_{\mathrm{S}},\epsilon)$-ULDP mechanism as a $(v,\epsilon)$-ULDP mechanism. 
We denote the class of all $(v,\epsilon)$-ULDP mechanisms by $\mathcal{Q}_{w,v,\epsilon}$. 
Also, a private estimation scheme $(Q,\hat{P})$ which satisfies $Q \in \mathcal{Q}_{w,v,\epsilon}$ is called a $(v,\epsilon)$\textbf{-ULDP scheme}.
Note that $(w,\epsilon)$-ULDP is equivalent to $\epsilon$-LDP.
Also, \eqref{eq:ULDP_non_sensitive} indicates that whenever the event $\{Y\in\mathcal{Y}_{\mathrm{I}}\}$ occurs, the server can perfectly infer the input $x \in \mathcal{X}_{\mathrm{N}}$ of a ULDP mechanism.
Thus, roughly speaking, an $(v,\epsilon)$-ULDP mechanism provides partly more accurate data than an $\epsilon$-LDP mechanism.

The \textbf{optimal privacy-utility tradeoff} (or, in short, \textbf{optimal PUT}) for a discrete distribution estimation system with $n$-clients under $(v,\epsilon)$-ULDP constraint is defined as 
\begin{equation}
    M^*_n(w,v,\epsilon) := \inf\limits_{Q \in \mathcal{Q}_{w,v,\epsilon}}\inf\limits_{\hat{P}}R_{n}(Q,\hat{P}).
\end{equation}
In \cite{murakami2019utility}, it was shown that $M^*_n(w,v,\epsilon) = \Theta(1/n)$.
Accordingly, the \textbf{asymptotically optimal PUT} is defined as 
\begin{equation}
    M^*(w,v,\epsilon):=\liminf\limits_{n\rightarrow \infty} n \cdot M^*_n(w,v,\epsilon),
\end{equation}
and the \textbf{asymptotic error} of a private estimation scheme $(Q,\hat{P})$ is defined as 
\begin{equation}
    R(Q,\hat{P}) := \limsup\limits_{n\rightarrow \infty} n \cdot R_{n}(Q,\hat{P}).
\end{equation}
We say that a $(v,\epsilon)$-ULDP scheme $(Q,\hat{P})$ is \textbf{asymptotically optimal} if
\begin{align}
    R(Q,\hat{P})  = M^*(w,v,\epsilon).
\end{align}
In this paper, we focus on characterizing the asymptotically optimal PUT $M^*(w,v,\epsilon)$ and simply refer to the asymptotically optimal PUT as the optimal PUT from now on.
Note that $M^*(w,w,\epsilon)$ is the optimal PUT under the $\epsilon$-LDP constraint, which has already been established in \cite{ye2018optimal}.
Hence, from now on, we assume $v<w$.
\begin{remark}\label{rmk:ULDPtoLDP}
    The following inequalities follow directly from the relation between LDP and ULDP:
    \begin{align}
        M^*_n(w,v,\epsilon) \geq M^*_n(v,v,\epsilon), M^*(w,v,\epsilon) \geq M^*(v,v,\epsilon). \label{eq:convBoundByLDPSensitive}
    \end{align}
    In Remark~\ref{rmk:ULDPtoLDPTight}, we show that the above bound is, in fact, tight under specific conditions.
\end{remark}
\section{Preliminaries and Related Works}\label{sec:preliminaries}
\subsection{Discrete Distribution Estimation under LDP \& ULDP}\label{subsec:prevSchemes}
Discrete distribution estimation under local differential privacy (LDP) has been extensively studied, and a variety of LDP schemes have been proposed such as randomized response (RR) \cite{warner1965randomized}, randomized aggregatable privacy-preserving ordinal response (RAPPOR) \cite{erlingsson2014rappor}, subset selection (SS) \cite{ye2018optimal}, Hadamard response (HR) \cite{acharya2019hadamard}, projective geometry response (PGR) \cite{feldman2022private}, optimized unary encoding (OUE) protocol \cite{wang2017locally}, and optimized local hasing (OLH) protocol \cite{wang2017locally}.
Among these, the SS is known to attain the optimal PUT $M^*(w,w,\epsilon)$ \cite{ye17_opt_PUT_l2}.
Recently, Park \emph{et al.} \cite{park2023block} introduced a class of block design schemes that subsumes RR, SS, HR, and PGR while also enabling the construction of new mechanisms that achieve optimal PUT with reduced communication costs.

Turning to the utility-optimized LDP (ULDP) setting, proposed ULDP schemes are typically derived by modifying classical LDP schemes to satisfy the ULDP requirements in Definition \ref{def:ULDP}.
Concretely, uRR and uRAP \cite{murakami2019utility} are derived from RR \cite{warner1965randomized} and RAPPOR \cite{erlingsson2014rappor}, respectively; uSS \cite{he2025addressing} is based on SS \cite{ye2018optimal}; uOUE and uOLH \cite{liu2024multidimensional, he2025addressing} are based on OUE and OLH \cite{wang2017locally}; and a variant of the Hadamard response \cite{acharya2020context}—referred to here as uHR—extends HR \cite{acharya2019hadamard} to the ULDP constraint. 
To characterize the optimal PUT under ULDP, Murakami \emph{et al.} \cite{murakami2019utility} utilized Remark~\ref{rmk:ULDPtoLDP} and known results in PUT over LDP in \cite{duchi2013local} to show that for $\epsilon\in[0,1]$, $M^*_n(w,v,\epsilon) = \Omega\left(\tfrac{v}{n\epsilon^2}\right)$.
Then, they stated that uRAP attains an estimation error of $\tfrac{4v}{n\epsilon^2}$ as $\epsilon \to 0$; hence, it is \emph{order-optimal} for $\epsilon\in[0,1]$.
Subsequently, Acharya \emph{et al.} \cite{acharya2019hadamard} refined the optimal PUT as $M^*_n(w,v,\epsilon) = \Theta\left(\tfrac{v^2}{n\epsilon^2}+\tfrac{w}{n\epsilon} \right)$ for the high privacy regime $\epsilon=O(1)$, and introduced the uHR scheme whose estimation error matches this rate up to constant factors.\footnote{They stated the sample complexity under the total variation loss (see \cite[Theorem 3]{acharya2020context}). We adapt the argument to obtain the corresponding MSE.}
Recently, \cite{he2025addressing} showed that $M^*_n(w,v,\epsilon) = \Theta\left(\tfrac{ve^\epsilon}{n(e^\epsilon-1)^2}\right)$ for $\epsilon \leq \log v$, and uSS, uOUE, and uOLH are \emph{order-optimal} in the regime $\epsilon \leq \log v$, which was also proved from Remark~\ref{rmk:ULDPtoLDP} and results in PUT over LDP in \cite{ye2018optimal}.
Furthermore, it was experimentally shown that uSS \cite{he2025addressing} achieves the lowest estimation error among previous schemes.

Despite these advances, the \emph{exact} optimal PUT under ULDP, in particular the determination of the leading constant, had remained an open problem until our work, where we establish a complete characterization.

\subsection{Asymptotic Cram\'er-Rao Lower Bound}\label{subsec:prevAsympCRLB}
A well-known result in asymptotic statistics \cite{polyanskiyInformationTheoryCoding2024, vaartAsymptoticStatistics1998, lecamAsymptoticsStatistics2000} is an asymptotic variant of the Cram\'er-Rao lower bound (CRLB), which we refer to as the \textbf{asymptotic Cram\'er-Rao lower bound} (asymptotic CRLB).
In \cite{ye17_opt_PUT_l2}, a key step in establishing the converse of the optimal PUT under LDP is to derive a certain uniform version of the asymptotic CRLB, suitable for distribution estimation under LDP. 
This subsection revisits the classical asymptotic CRLB and sketches the key ideas behind the uniform refinement of \cite{ye17_opt_PUT_l2}.
In Section \ref{subsec:UnifAsympCRLB}, we extend the uniform version of \cite{ye17_opt_PUT_l2} substantially, laying the groundwork for showing the converse of the optimal PUT under ULDP. 

The standard asymptotic CRLB \cite[Theorem 29.4]{polyanskiyInformationTheoryCoding2024} implies that for a \emph{fixed} conditional distribution $Q$ from $[w]$, we have
    \begin{align}
    \lim_{n \rightarrow \infty }n \cdot \inf_{\hat{P}} R_n(Q,\hat{P}) \geq \sup_{P \in \Delta_w^\circ}\trInv\left(J_{P,Q}\right), \label{eq:standAsympCRB}
    \end{align}
where $J_{P,Q}$ is the Fisher information matrix of $Q$ at $P$ (defined in Definition~\ref{def:FI}).
However, this cannot be directly applied to derive a lower bound on the optimal PUT under LDP, which involves the infimum over the privacy mechanisms $Q$ \emph{inside the limit over $n$}.
To circumvent this challenge, the authors of  \cite{ye17_opt_PUT_l2} first reduced the optimization over \emph{all $\epsilon$-LDP mechanisms} to the optimization over a specific class of $\epsilon$-LDP mechanisms, called the \emph{extremal $\epsilon$-LDP mechanisms} introduced in \cite{kairouz2016extremal}. 
Then, by refining the classical proof of asymptotic CRLB in \cite{vaartAsymptoticStatistics1998, lecamAsymptoticsStatistics2000}, commonly referred to as the \emph{local asymptotic normality (LAN)} argument, they derived a certain universal lower bound on $\inf_{\hat{P}} R_n(Q,\hat{P})$, holding for \emph{all extremal $\epsilon$-LDP mechanisms} $Q$. This gives 
\begin{align}
    M^*(w,w,\epsilon) \geq \sup_{P \in \Delta_w^\circ} \inf_{Q \in \mathcal{Q}_{w,\epsilon}^{\mathrm{E}}} \trInv\left(J_{P,Q}\right),\label{eq:LDPUnifAsympCRLB}
\end{align}
where $\mathcal{Q}_{w,\epsilon}^{\mathrm{E}}$ is the set of all extremal $\epsilon$-LDP mechanisms \cite{kairouz2016extremal}.
We note that the proof exploits the special structures of (extremal) LDP mechanisms, and therefore it is not straightforward to adopt the arguments of \cite{ye17_opt_PUT_l2} in the ULDP setup.

\subsection{Block Design Schemes}
\label{sec:BD}
Our proposed ULDP mechanism builds on the block design (BD)–based LDP mechanism in \cite{park2023block}, referred to as the BD  mechanism. 
This subsection reviews the BD mechanism with emphasis on elements relevant to the development of the proposed ULDP scheme in  Section~\ref{sec:achievabilityuBD}.
\begin{definition}[Block design \cite{ionin2006combinatorics}]\label{def:BD}
Let $v,b,r,k \in \mathbb{N}$ and $\lambda \in \mathbb{Z}_{\geq 0}$. A hypergraph $(\mathcal{V},\mathcal{E})$, that consists of a set of vertices $\mathcal{V}$ and a set of edges (or blocks) $\mathcal{E} \subseteq 2^{\mathcal{V}}$, is called a $(v,b,r,k,\lambda)$-block design if $|\mathcal{V}| = v$, $|\mathcal{E}| = b$, and it satisfies the following symmetries:
\begin{enumerate}
    \item $r$-regular: every vertex in $\mathcal{V}$ is contained in $r$ edges.
    \item $k$-uniform: $\forall e \in \mathcal{E}$, $|e| = k$.
    \item $\lambda$-pairwise balanced: every pair of vertices in $\mathcal{V}$ is contained in $\lambda$ edges.
\end{enumerate}
\end{definition}

Without loss of generality, we assume $\mathcal{V} = [|\mathcal{V}|]$.
The parameters $(v,b,r,k,\lambda)$ of a block design should satisfy the following relations.

\begin{lemma}\cite{ionin2006combinatorics}\label{bdlemma}
For any $(v,b,r,k,\lambda)$-block design, we have
\begin{equation}
    b \geq v \geq k, \quad b \geq r \geq \lambda,
\end{equation}
\begin{equation}
    bk = vr, \quad r(k-1) = \lambda(v-1).
\end{equation}

Also, if $v=1$, then $(b,r,k,\lambda)=(1,1,1,0)$.
\end{lemma}

\begin{remark}\label{rmk:CBD}
    For any given $v,k \in \mathbb{N}$ such that $k \leq v$, there exists a $\left(v,\binom{v}{k}, \binom{v-1}{k-1}, k, \binom{v-2}{k-2} \right)$-block design called the $(v,k)$-complete block design \cite{ionin2006combinatorics}
    (here, if $k=1$, we regard $\binom{v-2}{k-2}=0$). This is formed by letting $\mathcal{E}$ be the set of \emph{all} size $k$ subsets of $\mathcal{V}$.
\end{remark}

Given $(v,b,r,k,\lambda)$-block design, a BD mechanism regards input data as a vertex and outputs an edge, enforcing $\epsilon$-LDP by biasing towards edges that contain the input vertex.

\begin{definition}[Block design mechanism \cite{park2023block}] \label{def:BD_mech}
For given $\epsilon>0$ and a $(v,b,r,k,\lambda)$-block design $(\mathcal{X},\mathcal{Y})$, the corresponding  $(v,b,r,k,\lambda,\epsilon)$-block design mechanism  is the conditional distribution $Q: \mathcal{X} \rightarrow \mathcal{P}(\mathcal{Y})$ such that
\begin{equation}
    Q(y|x) = 
    \begin{cases}
        \frac{e^\epsilon}{r e^\epsilon + b - r} & (\text{if } x \in y)\\
        \frac{1}{r e^\epsilon + b - r} & (\text{if } x \notin y)
    \end{cases}
\end{equation}
\end{definition}

Park \emph{et al.} also proposed the corresponding sequence of unbiased estimators $\hat{P}$ for the BD mechanism \cite[Thm. 3]{park2023block}, and called such a pair $(Q,\hat{P})$ a $(v,b,r,k,\lambda,\epsilon)$-BD scheme.
For any $(v,\epsilon)$, it was shown that the estimation error of the BD scheme only depends on the uniformity parameter $k$. In \cite[Prop. 1 and Thm. 5]{park2023block}, an explicit characterization of optimal $k$ that minimizes the estimation error is provided and the BD scheme with optimal $k$ is shown to achieve the optimal PUT.

\begin{theorem}[\!\!{\cite[Prop. 1 and Thm. 5]{park2023block}}]\label{thm:BD_opt}
    Let $v \geq 2$ and let $(Q,\hat{P})$ be a $(v,b,r,k,\lambda,\epsilon)$-block design scheme.
    Its asymptotic error is
    \begin{align}
        R(Q,\hat{P})=R^\mathrm{BD}(v,k,\epsilon) := \frac{(v-1)^2(ke^\epsilon+v-k)^2}{vk(v-k)(e^\epsilon-1)^2}.\label{eq:defRBD}
    \end{align}
    Moreover, we have
    \begin{align}
        M^*(v,v,\epsilon) = \min_{k \in [v-1]}R^\mathrm{BD}(v,k,\epsilon),
    \end{align}
    and
    \begin{align}
        K^{*}(v,\epsilon) &:= \argmin\limits_{k\in[v-1]}R^{\mathrm{BD}}(v,k,\epsilon) \label{eq:defKStar1}
        \\& = \left\{k: E(v,k) \leq \epsilon \leq E(v,k-1) \right\},\label{eq:defKStar2}
    \end{align}
    where
    \begin{equation}\label{eq:E_vk}
        E(v,k)=\ln{\sqrt{\frac{(v-k)(v-k-1)}{k(k+1)}}}, \quad E(v,0):=\infty. 
    \end{equation}
\end{theorem}
By Theorem~\ref{thm:BD_opt}, it suffices to find a block design with $k=k^* \in K^{*}(v,\epsilon)$ to construct an optimal $\epsilon$-LDP scheme.
Remark~\ref{rmk:CBD} guarantees the existence of such a $(v,k^*,\epsilon)$-block design yielding a scheme equivalent to the optimal SS \cite{ye2018optimal}.

\section{Main Results}\label{sec:main}
Our main contribution is to \emph{completely} characterize the optimal PUT for the discrete distribution estimation under the ULDP constraint, i.e., the characterization of  $M^*(w,v,\epsilon)$ for all possible $(w,v,\epsilon)$. 
\begin{theorem}\label{thm:ULDP_opt_PUT}
    For $w > v \geq 1$ and $\epsilon > 0$, 
    \begin{align}
        M^*(w,v,\epsilon) = \sup_{\alpha \in [0,1]}\inf_{t \in \Delta_v}{M}(\alpha,t),\label{eq:ULDPOptPUT}
    \end{align}
    where $M(\alpha,t)=\sum_{i=1}^{3}M_i(\alpha,t)$ such that\footnote{
    If $v=1$, we regard $M_1(\alpha,t)=0$. If $v\geq 2$ and $t=\delta^{(v;v)}$, we regard $M_1(\alpha,t)=\infty$.
    }
    \begin{multline}
        M_1(\alpha,t)=\\
        \frac{(v-1)^2}{v(e^\epsilon-1)^2 \sum_{k=1}^{v}t_k \frac{k(v-k)}{(\alpha k(e^\epsilon-1)+v)(ke^\epsilon+v-k)}},
    \end{multline}
    \begin{align}
        M_2(\alpha,t)&=
        \frac{(w-v-1)(1-\alpha)}{(w-v)(e^\epsilon-1)\sum_{k=1}^{v}t_k \frac{k}{ke^\epsilon+v-k}},\\
        M_3(\alpha,t)&=
        \frac{w(1-\alpha)}{v(w-v)(e^\epsilon-1)\sum_{k=1}^{v} t_k \frac{k}{\alpha k(e^\epsilon-1)+v}}.
    \end{align}
    Here, $M(\alpha,t)$ is concave-convex and has a saddle point $(\alpha^*,t^*)$, which gives $\sup_{\alpha}\inf_{t} M(\alpha,t)=M(\alpha^*, t^*)$.\footnote{
    $(\alpha^*,t^*) \in [0,1] \times \Delta_v$ is a saddle point of $M$ if $M(\alpha^*,t^*) = \sup_{\alpha \in [0,1]} M(\alpha,t^*) = \inf_{t \in \Delta_v} M(\alpha^*,t)$  \cite{boyd_convex_2004}.}  
\end{theorem}
For Theorem \ref{thm:ULDP_opt_PUT}, the converse and the achievability parts are shown in Sections~\ref{sec:converse} and~\ref{sec:achievabilityuBD}, respectively. The concave-convex property of $M(\alpha, t)$ and the existence of a saddle point is proved in Appendix~\ref{app:objCvxandAttainSaddle}.

For converse, we first extend two previous results: (i) generalize the uniform asymptotic Cramér–Rao lower bound \cite{ye17_opt_PUT_l2} introduced in Section~\ref{subsec:prevAsympCRLB}, and (ii)  show that it suffices to consider only a newly defined class of \emph{extremal ULDP mechanisms} as optimal, in a manner analogous to the extremal LDP mechanisms in \cite{kairouz2016extremal}. Then, we introduce a novel argument of \emph{decomposing} the data distribution $P$ into three components: (i) the relative values of $P_x$'s over sensitive $x$, (ii) the relative values of $P_x$'s over non-sensitive $x$, and (iii) the total probability of sensitive data $P(\mathcal{X}_{\mathrm{S}})$.
Using this, we derive a lower bound on $M^*(w,v,\epsilon)$ which can be interpreted as the \emph{sum of CRLBs for estimating each component}. Each $M_i$ in Theorem \ref{thm:ULDP_opt_PUT} can be intuitively interpreted as the CRLB for estimating the $i$-th component of the data distribution $P$.

For achievability, we propose a class of \textbf{utility-optimized block design schemes} (uBD schemes).
Each uBD scheme is associated with parameters $(\alpha,t)$ appearing in Theorem~\ref{thm:ULDP_opt_PUT}, and we show that if $(\alpha^*,t^*)$ is a saddle point of $M$, then every uBD scheme with parameter $(\alpha,t)=(\alpha^*,t^*)$ is asymptotically optimal. 
The privacy mechanism $Q$ of a uBD scheme builds upon the block design mechanisms \cite{park2023block} in Definition~\ref{def:BD_mech}, with nontrivial modifications. The distinguishing idea is to exploit a \emph{mixture} of multiple BD mechanisms, rather than restricting to a single BD mechanism.
The parameter $t$ represents the mixture proportions, where $t_k$ denotes the proportion of a $k$-uniform BD mechanism. The estimator sequence is constructed based on two ideas. First, we adopt the idea, also utilized in the converse proof, of decomposing the data distribution $P$ into three components, and estimate each component separately.
Second, we propose a set of estimator sequences $\{\hat{P}^{(\alpha)}\}_{\alpha \in [0,1]}$, each of which attains the Cram\'er-Rao lower bound at a specific data distribution $P^{(\alpha)} \in \Delta_w$.
The supremum over $\alpha$ in Theorem~\ref{thm:ULDP_opt_PUT} reflects the intuition that for the worst-case optimality, we use the estimator sequence that attains the worst-case CRLB. 
To attain the CRLB, the idea of `score-based linear estimator' in \cite{namOptimalPrivateDiscrete2024} is utilized.

Theorem \ref{thm:ULDP_opt_PUT}
characterizes the optimal PUT in terms of a finite-dimensional concave–convex saddle-point optimization problem, which can be efficiently solved using existing optimization solvers. Moreover, the following theorem gives a closed-form solution to this optimization problem for some regime, which is proved in Appendix~\ref{app:proofClosedFormExpPreciseForm}. 
\begin{theorem}\label{thm:closedFormExpPreciseForm}
\,

    \textbf{Case (a):} Suppose that $(w,v,\epsilon)$ satisfies at least one of the following three conditions:
        \begin{enumerate}
            \item[(i)] $v=1$;
            \item[(ii)] $v \geq 2$ and $\epsilon \geq \ln\left(w-v+\sqrt{\dfrac{(w-1)(w-2)}{2}}\right)$;
            \item[(iii)] $v=2$ and $\epsilon \leq \ln\left(1+\sqrt{\dfrac{2(w-2)}{w-1}}\right)$;
        \end{enumerate}
    Then, the following $(\alpha^*,t^*)$ is a saddle point of $M$, i.e., $M^*(w,v,\epsilon)=M(\alpha^*,t^*)$:
    \begin{equation}
        \alpha^* = \left[\frac{v(e^\epsilon-1-w+v)}{w(e^\epsilon-1)}\right]_+, \quad t^* = \delta^{(1 ;v)},
    \end{equation}
    where $[z]_+:=\max(0, z)$.
        
    \textbf{Case (b):} Suppose that $(w,v,\epsilon)$ satisfies $v \geq 4$ and $\epsilon \leq \ln\sqrt{\tfrac{(v-1)(v-2)}{2}}$.
    Then, for each $k^* \in K^*(v,\epsilon) \cap [2:v-1]$, the following $(\alpha^*,t^*)$ is a saddle point of $M$, i.e., $M^*(w,v,\epsilon)=M(\alpha^*,t^*)$:
    \begin{align}
        \alpha^* = 1,\quad t^* = \delta^{(k^* ;v)}.
    \end{align}
    Here, $K^*(v,\epsilon)$ is defined in \eqref{eq:defKStar1}.\footnote{
    Note that $\ln\sqrt{\tfrac{(v-1)(v-2)}{2}}=E(v,1)$, where $E$ is in \eqref{eq:E_vk}. Hence $K^*(v,\epsilon) \cap [2:v-1] \neq \emptyset$.
    }
\end{theorem}

The closed-form expression in Theorem~\ref{thm:closedFormExpPreciseForm} allows us to verify the optimality or suboptimality of existing schemes, as presented in the following corollary.
\begin{corollary}\label{cor:uRROptNaive}
  If $(w,v,\epsilon)$ corresponds to Case (a) of Theorem~\ref{thm:closedFormExpPreciseForm}, uRR \cite{murakami2019utility} is asymptotically optimal. 
  \end{corollary}
 \begin{corollary}\label{cor:uSSStrictSuboptimal}
If  $(w,v,\epsilon)$ corresponds to Case (b) of Theorem~\ref{thm:closedFormExpPreciseForm} with strict inequality with respect to $\epsilon$, i.e., $v \geq 4$ and $\epsilon <\ln\sqrt{\tfrac{(v-1)(v-2)}{2}}$, the uSS scheme is suboptimal. In other words, for the uSS scheme $(Q_{\mathrm{uSS},k}, \hat{P}_{\mathrm{uSS},k})$ with parameter $k\in [v-1]$,  
    \begin{align}
        M^*(w,v,\epsilon)<\min_{k \in [1:v-1]} R(Q_{\mathrm{uSS},k}, \hat{P}_{\mathrm{uSS},k}).
    \end{align}
\end{corollary}

Corollary~\ref{cor:uRROptNaive} is a special case of Corollary~\ref{cor:uBDk2optimal} and Corollary~\ref{cor:uSSStrictSuboptimal} is proved in Appendix~\ref{app:compBwuSS}. Corollary~\ref{cor:uRROptNaive} follows from the observation that our uBD scheme with the optimal parameter coincides with uRR in the corresponding regime, resembling the fact that the optimal BD scheme in the low-privacy regime reduces to RR \cite{warner1965randomized}. 
Since the proposal of uRR in \cite{murakami2019utility}, this is the first proof of its optimality in a certain regime.
Corollary~\ref{cor:uSSStrictSuboptimal} shows that uSS scheme, which is the best known previous scheme, is suboptimal, thereby demonstrating a strict gap between uBD and the previously known ULDP schemes. 

We conclude this section with the following remark on the relation with the trivial converse bound in Remark~\ref{rmk:ULDPtoLDP}.
\begin{remark}\label{rmk:ULDPtoLDPTight}
    In the regime of Case (b) in Theorem~\ref{thm:closedFormExpPreciseForm}, it can be easily seen that
    \begin{align}
        M(\alpha^*,t^*)
        =R^{\mathrm{BD}}(v,k^*,\epsilon)=M^*(v,v,\epsilon),
    \end{align}
    where $R^{\mathrm{BD}}$ is in \eqref{eq:defRBD}. Thus, the converse bound in Remark~\ref{rmk:ULDPtoLDP} is tight for $v \geq 4$ and $\epsilon \leq \ln\sqrt{\tfrac{(v-1)(v-2)}{2}}$.
\end{remark}

\section{Converse}\label{sec:converse}
The goal of this section is to prove the converse part of Theorem \ref{thm:ULDP_opt_PUT}, that is, to show
\begin{align}
    M^*(w,v,\epsilon) \geq \sup_{\alpha \in [0,1]}\inf_{t \in \Delta_v} M(\alpha,t).
\end{align}
This section consists of three subsections. Section~\ref{subsec:UnifAsympCRLB} provides a generalized uniform version of asymptotic CRLB \cite{ye17_opt_PUT_l2}, and Section~\ref{subsec:extremalULDPMech} provides the notion of \emph{extremal ULDP mechanisms} analogous to extremal LDP mechanisms \cite{kairouz2016extremal}. Finally, Section~\ref{subsec:sepFISubspaces} finishes the converse proof by using a \emph{distribution decomposition} argument, which also plays an important role in the construction of proposed schemes.

\subsection{Uniform Asymptotic Cram\'er-Rao Lower Bound}\label{subsec:UnifAsympCRLB}
As a first step for the converse proof, we show the generalization of the uniform result in \cite{ye17_opt_PUT_l2}, which we call the \textbf{uniform asymptotic Cram\'er-Rao lower bound (uniform asymptotic CRLB)}. 
Our version of uniform asymptotic CRLB covers the general private discrete distribution estimation, not only private estimation under LDP or ULDP constraint, and thus we expect that it has possible applications to other private estimation problems.
Moreover, we provide a much shorter proof for the uniform asymptotic CRLB compared to \cite{ye17_opt_PUT_l2}.

First, we present the necessary definitions for CRLB.
\begin{definition}\label{def:scoreFunc}
    Let $w \in \mathbb{N}_{\geq 2}$ be given. Let $Q$ be a conditional distribution from $[w]$ to a finite set $\mathcal{Y}$, and $P \in \Delta_w^\circ$.
    For each $h \in \vec{\Delta}_w$, the \textbf{score} of $Q$ at $P$ in direction $h$ is a function $\eta_{P,Q;h}:\supp(Q) \rightarrow \mathbb{R}$, defined by
    \begin{align}
        \eta_{P,Q;h}(y) &= \frac{d}{d\delta}(\log Q_{P+\delta h}(y))|_{\delta=0}\\
        &= \langle \eta_{P,Q}(y), h \rangle.
    \end{align}
    Here, $\eta_{P,Q}:\supp(Q) \rightarrow \mathbb{R}^w$ is the \textbf{score vector}, given by
    \begin{align}
        (\eta_{P,Q}(y))_x = \frac{Q(y|x)}{Q_P(y)}, \quad \forall x \in [w]. \label{eq:defScoreFuncDef}
    \end{align}
\end{definition}

\begin{definition}\label{def:FI}
        The \textbf{Fisher information} of $Q$ at $P$ is a symmetric positive semidefinite bilinear map $\mathcal{J}_{P,Q}:\vec{\Delta}_w^2 \rightarrow \mathbb{R}$, defined as 
    \begin{align}
        \mathcal{J}_{P,Q}(h, h') = \cov_{Y \sim Q_P}(\eta_{P,Q;h}(Y), \eta_{P,Q;h'}(Y)).
    \end{align}
    If an orthonormal basis $\mathcal{B}=(h^{(1)},\cdots,h^{(w-1)})$ of $\vec{\Delta}_w$ is given, the \textbf{Fisher information matrix} with respect to such basis is the $(w-1) \times (w-1)$ matrix $J_{P,Q;\mathcal{B}}$ such that $(J_{P,Q;\mathcal{B}})_{i,j} = \mathcal{J}_{P,Q}(h^{(i)}, h^{(j)})$. We omit $\mathcal{B}$ from the subscript if it is clear from the context.
\end{definition}

\begin{remark}\label{rmk:FIMDiffBases}
    Fisher information matrices with respect to different orthonormal bases are conjugate via an orthogonal matrix. Therefore, many expressions involving trace of a function of Fisher information matrix (such as $\trInv(J_{P,Q})$ in \eqref{eq:standAsympCRB}) are invariant under the choice of basis. If so, we may omit the subscript $\mathcal{B}$ without ambiguity.
\end{remark}

Next, we present the notion of degradedness \cite{blackwellEquivalentComparisonsExperiments1953}, which gives a systematic way of reducing the space of privacy mechanisms to be optimized, as \cite{ye17_opt_PUT_l2} reduced the optimization over all LDP mechanisms to that over extremal LDP mechanisms.
Intuitively speaking, we say that $Q$ degrades $\tilde{Q}$ if $\tilde{Q}$ can be simulated by applying some post-processing to the output of $Q$. The formal definitions are as follows.

\begin{definition}
For finite sets $\mathcal{X},\mathcal{Y},\tilde{\mathcal{Y}}$ and two conditional distributions $Q:\mathcal{X} \rightarrow \mathcal{P}(\mathcal{Y})$ and $\tilde{Q}:\mathcal{X} \rightarrow \mathcal{P}(\tilde{\mathcal{Y}})$, we say $Q$ \textbf{degrades} $\tilde{Q}$, and denote $Q \sqsupseteq \tilde{Q}$, if there exists a conditional distribution $T:\mathcal{Y} \rightarrow \mathcal{P}(\mathcal{\tilde{Y}})$ such that $\tilde{Q}(\tilde{y}|x)=\sum_{y \in \mathcal{Y}}Q(y|x)T(\tilde{y}|y)$ for every $x \in \mathcal{X}$ and $\tilde{y} \in \tilde{\mathcal{Y}}$. 
\end{definition}
\begin{definition}
    For a finite set $\mathcal{X}$ and two classes $\mathcal{Q},\tilde{\mathcal{Q}}$ of conditional distributions from $\mathcal{X}$ to some finite set, we say $\mathcal{Q}$ \textbf{degrades} $\tilde{\mathcal{Q}}$, and denote $\mathcal{Q} \sqsupseteq \tilde{\mathcal{Q}}$ if for every $\tilde{Q} \in \tilde{\mathcal{Q}}$, there exists $Q \in \mathcal{Q}$ which degrades $\tilde{Q}$.
\end{definition}

Now, we state the uniform asymptotic CRLB.
\begin{theorem}[Uniform Asymptotic Cram\'er-Rao Lower Bound]\label{thm:unifAsympCRLB}
Let $w\in \mathbb{N}_{\geq 2}$ be given, and let $\mathcal{X}=[w]$. Let $\tilde{\mathcal{Q}}$ be a class of conditional distributions from $\mathcal{X}$ to some finite set. Suppose that there exist a \emph{fixed} finite set $\mathcal{Y}$ and a \emph{compact} subset $\mathcal{Q}$ of conditional distributions from $\mathcal{X}$ to $\mathcal{Y}$, which degrades $\tilde{\mathcal{Q}}$. Then, we have
\begin{equation}
    \lim_{n \rightarrow \infty} \inf\limits_{\tilde{Q} \in \tilde{\mathcal{Q}}}\inf\limits_{\hat{P}}n\cdot R_{n}(\tilde{Q},\hat{P}) \geq 
    \sup_{P \in \Delta_w^\circ} \inf_{Q \in \mathcal{Q}} \trInv\left(J_{P,Q}\right).
\end{equation}
\end{theorem}
We prove Theorem~\ref{thm:unifAsympCRLB} by refining the proof of asymptotic CRLB in \cite{polyanskiyInformationTheoryCoding2024}.
\begin{proof}[Proof of Theorem \ref{thm:unifAsympCRLB}]
    It suffices to show that for each fixed $P \in \Delta_w^\circ$, we have
    \begin{align}
        \lim_{n \rightarrow \infty} \inf\limits_{\tilde{Q} \in \tilde{\mathcal{Q}}}\inf\limits_{\hat{P}}n\cdot R_{n}(\tilde{Q},\hat{P}) \geq \inf_{Q \in \mathcal{Q}} \trInv\left(J_{P,Q}\right). \label{eq:unifAsympCRLBFixedDist}
    \end{align}
    Hence, from now on, we fix $P \in \Delta_w^\circ$.
    
    \paragraph{Instantiation of Proof in \cite{polyanskiyInformationTheoryCoding2024}}
    The first half of the proof is the same as the proof of asymptotic CRLB in \cite{polyanskiyInformationTheoryCoding2024}. For the detailed reason for each step, refer to \cite[Section 29.1]{polyanskiyInformationTheoryCoding2024}.
    
    Since $P \in \Delta_w^\circ$, there exists $D>0$ such that whenever $h \in \vec{\Delta}_w$, $\Vert h \Vert \leq D$, we have $P+h \in \Delta_w^\circ$.
    Let us choose an orthonormal basis $\mathcal{B}=(h^{(1)},\cdots,h^{(w-1)})$ of $\vec{\Delta}_w$, and for each $\theta \in \mathbb{R}^{w-1}$, $\Vert \theta \Vert \leq R$, let $\tilde{P}^{(\theta)}:= P+\sum_{i=1}^{w-1} \theta_i h^{(i)} \in \Delta_w^\circ$. 
    Then, whenever $n \geq (w-1)^2/D^4$, we can bound the worst-case MSE by the average MSE,
    \begin{align}
        R_n(\tilde{Q},\hat{P}) \geq \mathbb{E}_{\Theta \sim \phi_n}\left[R_n(\tilde{Q},\hat{P};\tilde{P}^{(\Theta)})\right],
    \end{align}
    where under $\Theta \sim \phi_n$, each of $\Theta_1,\cdots,\Theta_{w-1}$ is i.i.d., supported on $[-n^{-1/4},n^{-1/4}]$, and following the PDF $f(z)=n^{1/4}\cos^2\left(\frac{n^{1/4}\pi z}{2}\right)$.\footnote{
        This PDF is shown \cite{polyanskiyInformationTheoryCoding2024,barnesFisherInformationLocal2020} to minimize the `prior Fisher information' over all PDFs supported on $[-n^{-1/4},n^{-1/4}]$.
    }
    Then, we apply another variant of CRLB, called \textbf{Bayesian CRLB}, resulting that whenever $Q \sqsupseteq \tilde{Q}$ and $n \geq (w-1)^2/D^4$, we have
    \begin{multline}
        \inf_{\hat{P}} \mathbb{E}_{\Theta \sim \phi_n}\left[R_n(\tilde{Q},\hat{P};\tilde{P}^{(\Theta)})\right] \geq\\
        \trInv(n\cdot \mathbb{E}_{\Theta \sim \phi_n}[J_{{\tilde{P}^{(\Theta)}}, Q}] + \pi^2 n^{1/2} \mathbf{I}_{w-1}). \label{eq:BCRLBResult}
    \end{multline}

    \paragraph{New Techniques for the Uniform Result}
    After that, to derive the uniform result, we establish a technical proposition which shows that $\{P \mapsto J^{P,Q}\}_{Q}$ is \emph{equicontinuous} in some sense. The formal statement is in Proposition~\ref{prop:FIBwTwoDistIneq} in Appendix~\ref{appsubsec:proofPropFIBwTwoDistIneq}.
    Using this proposition and the $n^{-1/4}$ order of the diameter of the support of $\Theta$, we derive in Appendix \ref{appsubsec:proofFIExpectationUnifBound} that there exist constants $C>0$ and $N \in \mathbb{N}$ (which only depend on $(w,P)$) such that for every $n \geq N$ and every conditional distribution $Q$ from $[w]$, we have
    \begin{align}
        \mathbb{E}_{\Theta \sim \phi_n}[J_{{\tilde{P}^{(\Theta)}}, Q}] \preceq (1+ C n^{-1/4})J_{P, Q}.\label{eq:FIExpectationUnifBound}
    \end{align}
    Using a standard fact \cite{bikchentaevTraceInequalitiesMatrices2024} that $A \preceq B$ implies $A^{-1} \succeq B^{-1}$ for positive definite matrices $A,B$, we obtain that
    \begin{equation}
        \inf\limits_{\hat{P}} R_{n}(\tilde{Q},\hat{P}) \geq \frac{1}{n}\overline{R}_n(Q) \label{eq:unifAsympCRLBFixedDistFiniteN}
    \end{equation}
    for every $Q \sqsupseteq \tilde{Q}$, where
    \begin{multline}
        \overline{R}_n(Q) =
        \trInv((1+ C n^{-1/4})J_{P, Q} + \pi^2 n^{-1/2} \mathbf{I}_{w-1}).
    \end{multline}
    Thus, we have
    \begin{align}
        \lim_{n \rightarrow \infty} \inf\limits_{\tilde{Q} \in \tilde{\mathcal{Q}}}\inf\limits_{\hat{P}}n\cdot R_{n}(\tilde{Q},\hat{P}) \geq \lim_{n \rightarrow \infty}\inf_{Q \in \mathcal{Q}} \overline{R}_n(Q). \label{eq:UnifAsympCRLBLimit}
    \end{align}
    
    Note that $\overline{R}_n(Q) \uparrow \trInv\left(J_{P,Q}\right)$ as $n \rightarrow \infty$. 
    Using compactness of $\mathcal{Q}$ and Dini's theorem \cite[Thm. 7.13]{rudin1976principles}, we show in Appendix~\ref{appsubsec:proofPropChangeLimitAndInfCompact} that
    \begin{align}
        \lim_{n \rightarrow \infty}\inf_{Q \in \mathcal{Q}} \overline{R}_n(P,Q) = \inf_{Q \in \mathcal{Q}} \trInv\left(J_{P,Q}\right). \label{eq:changeLimitAndInfCompact}
    \end{align}
    This finishes the proof of Theorem~\ref{thm:unifAsympCRLB}.
\end{proof}

\subsection{Extremal ULDP Mechanisms}\label{subsec:extremalULDPMech}
As a next step, to utilize the uniform asymptotic CRLB in the ULDP setup, we find a certain \emph{compact} set of conditional distributions which degrades the class of all ULDP mechanisms. In short, we find such a set analogous to the set of extremal LDP mechanisms \cite{kairouz2016extremal}, which we call the set of \emph{extremal ULDP mechanisms}.

\begin{definition}\label{def:ULDPExtremal}
Let $1 \leq v < w$, $\epsilon>0$ be given, and $\mathcal{X} = [w]$, $\mathcal{X}_{\mathrm{S}} = [v]$. 
A conditional distribution $Q$ is called an \textbf{extremal $(v,\epsilon)$-ULDP mechanism} if the following holds:
\begin{itemize}
    \item $Q:\mathcal{X} \rightarrow \mathcal{P}(\mathcal{Y}_{\mathrm{P}} \sqcup \mathcal{Y}_{\mathrm{I}})$, where $\mathcal{Y}_{\mathrm{P}} = 2^{[v]} \backslash \{\emptyset\}$ and $\mathcal{Y}_{\mathrm{I}} = \{\{x\}: x \in [v+1:w]\}$.
    \item There exists $\gamma:\mathcal{Y}_{\mathrm{P}} \rightarrow \mathbb{R}_{\geq 0}$ such that if $y \in \mathcal{Y}_{\mathrm{P}}$,
    \begin{align}
        Q(y|x) 
        &=\begin{cases}
            \gamma(y) e^\epsilon & (\text{if } x \in y)\\
            \gamma(y)            & (\text{if } x \notin y)
        \end{cases}.
    \end{align}
    \item If $y \in \mathcal{Y}_{\mathrm{I}}$,
    \begin{align}
        Q(y|x) = \left(1-\sum_{y' \in \mathcal{Y}_{\mathrm{P}}} \gamma(y')\right)\mathbbm{1}(x \in y).
    \end{align}
\end{itemize}
We define $\mathcal{Q}_{w,v,\epsilon}^{\mathrm{E}}$ to be the set of all extremal $(v,\epsilon)$-ULDP mechanisms.
\end{definition}
We can equivalently describe an extremal ULDP mechanism $Q$ by its stochastic matrix as follows. The visualization of the stochastic matrix is in Figure~\ref{fig:extremal}. We identify $\mathcal{Y}_{\mathrm{P}} \simeq [1:2^{v}-1]$ via a bijection $y \leftrightarrow \sum_{i=1}^{v}2^{i-1} y_i$, and identify $\mathcal{Y}_{\mathrm{I}} \simeq [2^v:2^v+w-v-1]$ via $\{x\}\leftrightarrow 2^v+x-v-1$. First, we bring the \textbf{staircase pattern matrix} in \cite{kairouz2016extremal}:
\begin{equation}
        S^{(v)} := \mathbf{1}_{v \times 2^v-1} + (e^\epsilon-1) \mathrm{BIN}^{(v)},
    \end{equation}
where $\mathrm{BIN}^{(v)} \in \{0,1\}^{v \times (2^v-1)}$ is the matrix whose $y$'th column corresponds to the binary representation of $y$.
Then, we set $Q_{\mathcal{X}_{\mathrm{S}},\mathcal{Y}_{\mathrm{P}}} = S^{(v)} \diag(\gamma)$, which is precisely an extremal $\epsilon$-LDP mechanism \cite{kairouz2016extremal}.
Next, we set $Q_{\mathcal{X}_{\mathrm{N}},\mathcal{Y}_{\mathrm{P}}}=\mathbf{1}_{(w-v)\times (2^v-1)} \diag(\gamma)$. 
Intuitively, for given $Q_{\mathcal{X}_{\mathrm{S}},\mathcal{Y}_{\mathrm{P}}}$, we set each element of $Q_{\mathcal{X}_{\mathrm{N}},\mathcal{Y}_{\mathrm{P}}}$ to be the minimum possible value while satisfying the ULDP constraint, so that the probability of returning invertible data given non-sensitive data is maximized.
Finally, we set $Q_{\mathcal{X}_{\mathrm{N}},\mathcal{Y}_{\mathrm{I}}}$ to be a constant multiple of an identity matrix, where the constant $f(\gamma)=1-\sum_{y' \in \mathcal{Y}_{\mathrm{P}}} \gamma(y')$ is determined by $\gamma$ so that for each $x \in \mathcal{X}_{\mathrm{N}}$, the $x$'th row of $Q$ is summed up to one. Also, we set $Q_{\mathcal{X}_{\mathrm{S}},\mathcal{Y}_{\mathrm{I}}}$ to be the all-zero matrix.
{
\begin{figure}[h!]
    \centering
    \begin{tikzpicture}
\node (S) {
$\displaystyle
S^{(3)} =
\begin{bNiceMatrix}[name=S]
  e^\epsilon & 1 & e^\epsilon & 1 & e^\epsilon & 1 & e^\epsilon \\
  1 & e^\epsilon & e^\epsilon & 1 & 1 & e^\epsilon & e^\epsilon \\
  1 & 1 & 1 & e^\epsilon & e^\epsilon & e^\epsilon & e^\epsilon
\end{bNiceMatrix}
$};
    \node[below = 5pt of S.south] (Q) {
    \begingroup
        $\displaystyle
        Q=
        \begin{bNiceArray}[first-row,first-col,margin]{c|c}
        {}& {\mathcal{Y}_\mathrm{P}} &{\mathcal{Y}_\mathrm{I}}  \\
        {\mathcal{X}_\mathrm{S}}& {S^{(v)}\operatorname{diag}(\gamma)} & {\mathbf{0}_{v\times(w-v)}} \\\hline
        {\mathcal{X}_\mathrm{N}} &{\mathbf{1}_{(w-v)\times(2^v-1)}\operatorname{diag}(\gamma)}& {f(\gamma)\mathbf{I}_{w-v}}
        \end{bNiceArray}
        $
    \endgroup
    };
        
\end{tikzpicture}
    \caption{Visualization of staircase pattern matrix and extremal $(v,\epsilon)$-ULDP mechanism.}
    \label{fig:extremal}
\end{figure}
}

Clearly, every extremal $(v,\epsilon)$-ULDP mechanism is a $(v,\epsilon)$-ULDP mechanism. Also, note that among previously proposed mechanisms, uRR \cite{murakami2019utility} and uHR \cite{acharya2020context} are equivalent to some extremal ULDP mechanisms.\footnote{
Here, the equivalence is up to relabeling $\mathcal{Y}$ and the removal / addition of all zero columns of $Q$.
}
It turns out that the set of all extremal ULDP mechanisms degrades the class of all ULDP mechanisms.

\begin{theorem}\label{thm:ULDPExtremal}
For every $1 \leq v < w$ and $\epsilon>0$, we have $\mathcal{Q}_{w,v,\epsilon}^{\mathrm{E}} \sqsupseteq \mathcal{Q}_{w,v,\epsilon}$.
\end{theorem}
The proof is similar to the proof of degradation of LDP mechanisms by extremal LDP mechanisms in \cite{namQuantumAdvantagePrivate2025, amorinoFactorizationExtremalPrivacy2025}, and we present the proof in Appendix~\ref{app:proofULDPExtremal}. 

To combine Theorems~\ref{thm:unifAsympCRLB} and \ref{thm:ULDPExtremal}, it remains to show that $\mathcal{Q}_{w,v,\epsilon}^{\mathrm{E}}$ is compact. 
It is straightforward to see that the given $\gamma:\mathcal{Y}_{\mathrm{P}} = 2^{[v]} \backslash \{\emptyset\} \rightarrow \mathbb{R}_{\geq 0}$ induces a valid extremal $(v,\epsilon)$-ULDP mechanism (satisfying $\sum_{y} Q(y|x)=1$ for all $x \in \mathcal{X}$) if and only if
\begin{align}
    \sum_{y \in \mathcal{Y}_{\mathrm{P}}} \gamma(y) (1+(e^\epsilon-1)\mathbbm{1}(x \in y)) = 1, \quad \forall x \in [v].
     \label{eq:ULDPGammaCond}
\end{align}
Clearly, the set of all $\gamma$ satisfying \eqref{eq:ULDPGammaCond} is compact, and hence $\mathcal{Q}_{w,v,\epsilon}^{\mathrm{E}}$ is compact.
Thus, by Theorems~\ref{thm:unifAsympCRLB} and \ref{thm:ULDPExtremal}, we obtain
\begin{align}
    M(w,v,\epsilon) \geq \sup_{P \in \Delta_w^\circ} \inf_{Q \in \mathcal{Q}_{w,v,\epsilon}^{\mathrm{E}}} \trInv\left(J_{P,Q}\right). \label{eq:unifAsympCRLBULDPApp}
\end{align}

\subsection{Distribution Decomposition for Tighter Bound}\label{subsec:sepFISubspaces}
Up to this point, we have established the converse bound by generalizing existing results.
Now, we present our prominent idea of \emph{decomposing} the data distribution, which not only lies behind the converse proof, but also plays an important role in the construction of our proposed schemes in Section~\ref{sec:achievabilityuBD}.
Based on this idea, we present the complete converse proof at the end of this subsection.

To formalize this idea, we introduce the following three mutually orthogonal subspaces of $\mathbb{R}^w$ which span the direction space $\vec{\Delta}_w$ in \eqref{eq:directSpaceDef},
\begin{align}
    \mathcal{H}_{1} &:=\{h: h_x=0 \text{ for }x \in [v+1:w], \sum_{x=1}^{v} h_x = 0\}, \label{eq:TanSpaceSubspace1}\\
    \mathcal{H}_{2} &:=\{h: h_x=0 \text{ for }x \in [v], \sum_{x=v+1}^{w} h_x = 0\},\label{eq:TanSpaceSubspace2}\\
    \mathcal{H}_{3} &:=\mathrm{span}\left(\begin{bmatrix}
        (w-v)\mathbf{1}_{v}\\
        -v \mathbf{1}_{w-v}
    \end{bmatrix}\right).\label{eq:TanSpaceSubspace3}
\end{align}
Let $d_i:=\dim(\mathcal{H}_i)$, that is, $(d_1,d_2,d_3)=(v-1,w-v-1,1)$.
Also, let $\Pi_{i}:\mathbb{R}^w \rightarrow \mathcal{H}_{i}$ be the orthogonal projection onto $\mathcal{H}_{i}$.
Then, intuitively, $\Pi_1(P)$ and $\Pi_2(P)$ represent the relative values of $P_x$'s over sensitive $x$ and that over non-sensitive $x$, respectively. Also, $\Pi_3(P)$ designates the probability of sensitive data $P(\mathcal{X}_{\mathrm{S}})$, as it is easy to see that
\begin{align}
    \Pi_3(P) = \frac{wP(\mathcal{X}_{\mathrm{S}})-v}{vw(w-v)}\begin{bmatrix}
        (w-v)\mathbf{1}_{v}\\
        -v \mathbf{1}_{w-v}
    \end{bmatrix}.
\end{align}
Additionally, we define the class of distributions $P^{(\alpha)} \in \Delta_w$, $\alpha \in [0,1]$, by the mixture of the uniform distribution on $\mathcal{X}_{\mathrm{S}}$ and the uniform distribution on $\mathcal{X}_{\mathrm{N}}$, with a rate $\alpha:1-\alpha$. More precisely, we have
\begin{align}
    P^{(\alpha)}_x = \begin{cases}
        \frac{\alpha}{v} & (\text{if }x \in \mathcal{X}_{\mathrm{S}})\\
        \frac{1-\alpha}{w-v} & (\text{if }x \in \mathcal{X}_{\mathrm{N}})
    \end{cases}.
    \label{eq:DefPAlpha}
\end{align}
The key idea for the converse proof is to give a lower bound on the optimal PUT, which can be intuitively interpreted as the \emph{sum of CRLBs for estimating $\Pi_i(P)$ over $i=1,2,3$}. 
Building on this idea, we now present the formal converse proof.
\begin{proposition}\label{prop:ConvPart}
For every $w>v \geq 1$ and $\epsilon>0$, we have
\begin{align}
    M^*(w,v,\epsilon) \geq \sup_{\alpha \in [0,1]}\inf_{t \in \Delta_v} M(\alpha,t).
\end{align}
\end{proposition}
\begin{proof}
First, we further bound \eqref{eq:unifAsympCRLBULDPApp} by replacing the supremum over all $P\in \Delta_w^\circ$ with the supremum over $P^{(\alpha)}$, $\alpha \in (0,1)$, obtaining
\begin{align}
    M^*(w,v,\epsilon) \geq \sup_{\alpha \in (0,1)} \inf_{Q \in \mathcal{Q}_{w,v,\epsilon}^{\mathrm{E}}} \trInv\left(J_{P^{(\alpha)}, Q}\right). \label{eq:ULDPUnifCRBLimitToPAlpha}
\end{align}
We choose an orthonormal basis $\mathcal{B}=\{h^{(1)},\cdots,h^{(w-1)}\}$ of $\mathcal{H}$, such that the first $(v-1)$ vectors form a basis of $\mathcal{H}_1$, the next $(w-v-1)$ vectors form a basis of $\mathcal{H}_2$, and the last vector $h^{(w-1)}$ is the unit vector in $\mathcal{H}_3$.

Now, let $Q \in \mathcal{Q}_{w,v,\epsilon}^{\mathrm{E}}$ and $\alpha \in (0,1)$ be given. Let $B_1, B_2, B_3$ be the principal submatrices of $J_{P^{(\alpha)},Q;\mathcal{B}}$, where $B_1$ is formed by selecting the first $(v-1)$ rows and columns, $B_2$ by selecting the next $w-v-1$ rows and columns, and $B_3$ by selecting the last row and column.
Using the block matrix inversion formula \cite{boyd_convex_2004}, we derive in Appendix~\ref{app:proofCRLBDecompIntoSubspaces} that
\begin{align}
   \trInv\left(J_{P^{(\alpha)}, Q}\right) \geq \sum_{i=1}^{3} \trInv\left(B_i\right). \label{eq:CRLBDecompIntoSubspaces}
\end{align}
By the arithmetic mean-harmonic mean inequality, we have $\trInv\left(B_i\right) \geq \frac{d_i^2}{\tr(B_i)}$. 
Also, we derive in Appendix~\ref{app:proofConvPart} that
\begin{align}
    \tr(B_i) = \frac{d_i^2}{M_i(\alpha, t(Q))} , \quad \forall i\in\{1,2,3\},  \label{eq:TraceFIRestPropToJi}
\end{align}
where $t(Q) \in \Delta_v$ is the PMF of $|Y'|$ when $Y'\sim Q_{P^{(1)}}$ (recall that in an extremal ULDP mechanism, $Y$ is a subset of $\mathcal{X}$).
Thus, we obtain
\begin{align}
    \trInv\left( J_{P^{(\alpha)}, Q}\right) \geq \sum_{i=1}^{3}{M_i(\alpha,t(Q))}=M(\alpha,t(Q)).
\end{align}
Since this holds for every $Q \in \mathcal{Q}_{w,v,\epsilon}^{\mathrm{E}}$ and $\alpha \in (0,1)$, we have
\begin{align}
    M^*(w,v,\epsilon) \geq \sup_{\alpha \in (0,1)}\inf_{t \in \Delta_v} M(\alpha,t).
\end{align}
By continuity of $M$ and compactness of $\Delta_v$, we can replace $\sup\limits_{\alpha \in (0,1)}$ by $\sup\limits_{\alpha \in [0,1]}$.
This ends the converse proof.
\end{proof}

\section{Achievability: Utility-Optimized Block Design Scheme}\label{sec:achievabilityuBD}
\begin{figure*}
    \centering
    \includegraphics[width = .9\linewidth]{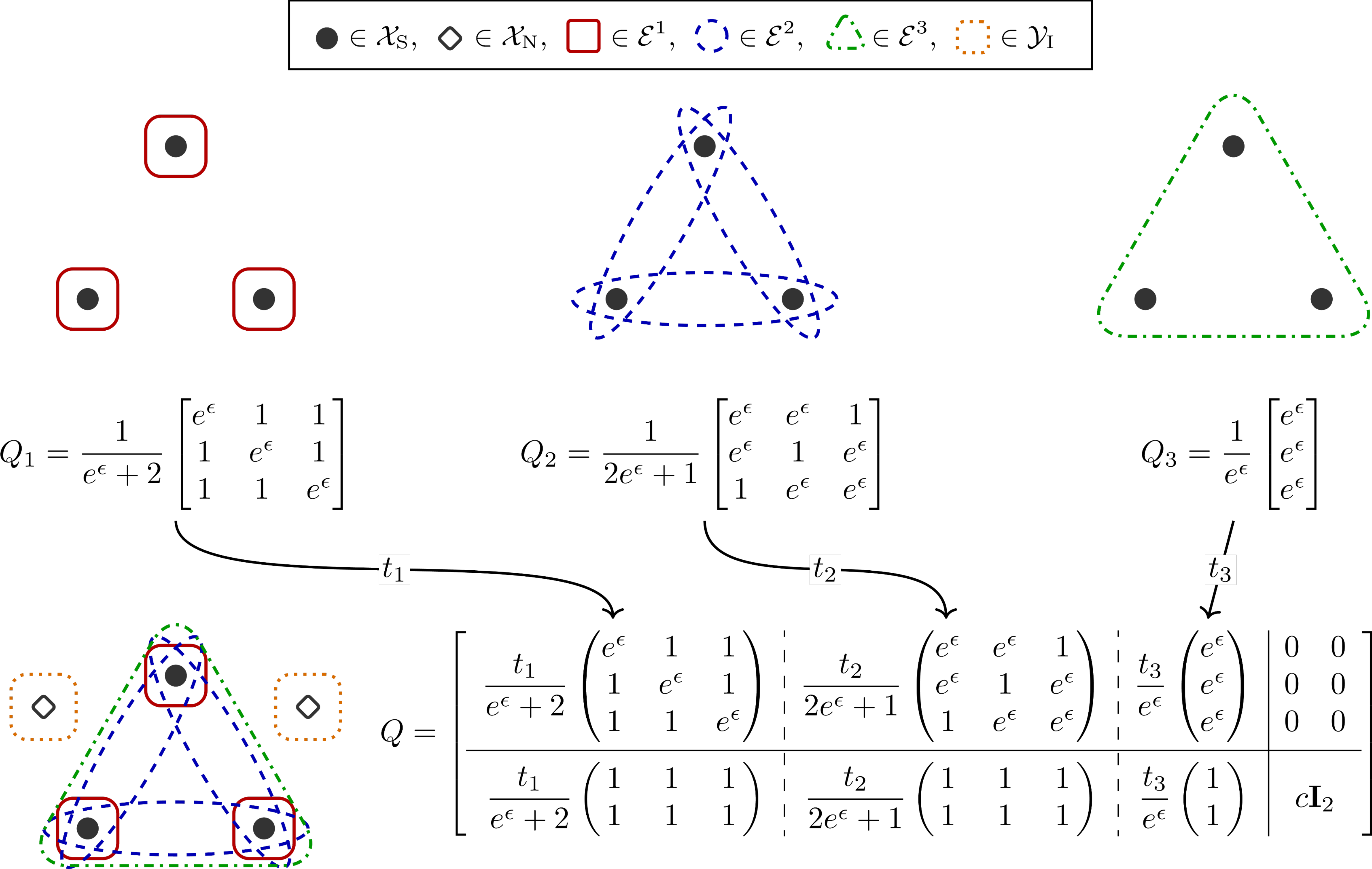}
    \caption{Visualization of $(5,3,\epsilon,t)$-uBD mechanism.}
    \label{fig:uBDvisual}
\end{figure*}
In this section, we propose a new class of ULDP schemes called \textbf{utility-optimized block design (uBD) schemes}. After that, we show that an optimized uBD scheme is asymptotically optimal.

\subsection{Construction of uBD Schemes}
\paragraph{Privacy mechanism}
We first propose a \textbf{uBD mechanism} which perturbs the data.
To construct a uBD mechanism, we first prepare a collection of block design mechanisms $(Q^{(k)}:[v] \rightarrow \mathcal{P}(2^{[v]}))_{k \in [v]}$, where each of $Q^{(k)}$ has the uniformity parameter $k$. This guarantees that the output spaces of $Q^{(k)}$'s are disjoint. Also, we prepare a mixture proportion $t \in \Delta_v$.
Then, a uBD mechanism $Q$ is constructed as follows. 
The visualization of $Q$ is in Figure~\ref{fig:uBDvisual}.
If a sensitive data $x \in \mathcal{X}_{\mathrm{S}}$ is given as the input of a uBD mechanism $Q$, then a uBD mechanism randomly chooses a block design mechanism $Q^{(k)}$, each with probability $t_k$, respectively, and simulates the chosen $Q^{(k)}$ to generate its output.
Hence, in terms of a stochastic matrix, $Q_{\mathcal{X}_{\mathrm{S}}, \mathcal{Y}_{\mathrm{P}}}$ is the horizontal stack of the stochastic matrices of $Q^{(k)}$ multiplied by $t_k$.
After that, we set $Q_{\mathcal{X}_{\mathrm{N}}, \mathcal{Y}}$ so that the resultant $Q$ becomes an \emph{extremal $(v,\epsilon)$-ULDP mechanism}, which is defined in Definition~\ref{def:ULDPExtremal}.
That is, for given $x \in \mathcal{X}_{\mathrm{N}}$, $Q$ either outputs one of the edges from the support of some $Q^{(k)}$, or directly outputs its input $x$ in the form of an edge $\{x\}$.
Here, the output distribution $Q(\cdot|x)$ is suitably determined to maximize the probability that a mechanism outputs $\{x\}$ while satisfying the $(v,\epsilon)$-ULDP constraint.
The precise definition of a uBD mechanism is as follows.

\begin{definition}\label{def:uBDMech}
Let $w > v \geq 1$, $\mathcal{X} = [w]$, $\mathcal{X}_{\mathrm{S}} = [v]$, $\epsilon >0$. Suppose that the following are given:
\begin{itemize}
    \item $t \in \Delta_v$,
    \item A collection of sets $\left(\mathcal{E}_k\right)_{k \in [v]}$, $\mathcal{E}_k \subset 2^{\mathcal{X}_{\mathrm{S}}}$, such that for each $k$, $\left(\mathcal{X}_{\mathrm{S}}, \mathcal{E}_k\right)$ is a $(v,b_k, r_k, k, \lambda_k)$-block design for some $(b_k, r_k, \lambda_k)$.
\end{itemize}
Then, the corresponding \textbf{$(w,v,\epsilon,t)$-uBD mechanism} is defined as an extremal $(v,\epsilon)$-ULDP mechanism as in Definition \ref{def:ULDPExtremal}, which is induced from the following $\gamma:\mathcal{Y}_{\mathrm{P}} \rightarrow \mathbb{R}_{\geq 0}$:
\begin{align}
    \gamma(y) = t_k \cdot \frac{\mathbbm{1}(y \in \mathcal{E}_k)}{r_k (e^\epsilon-1) + b_k}, \quad \forall y \text{ such that }|y|=k.
\end{align}
\end{definition}
By the regularity of block design, it is straightforward to check that this $\gamma$ satisfies the necessary and sufficient condition to induce an extremal ULDP mechanism in \eqref{eq:ULDPGammaCond}.
Note that for every $(w,v,\epsilon,t)$, there exists a $(w,v,\epsilon,t)$-uBD mechanism, by setting each of $(\mathcal{X}_{\mathrm{S}}, \mathcal{E}_k)$ to be the complete block design in Remark \ref{rmk:CBD}.
Also, to specify a uBD mechanism, it suffices to designate $\mathcal{E}_k$ only for $k \in [v]$ such that $t_k>0$, not necessary for all $k \in [v]$.

\paragraph{Estimator} 
Now, we propose a class of estimators for a uBD mechanism.
The main ideas for our proposed estimators are as follows.
\begin{enumerate}
    \item[(i)] First, as we derive the lower bound based on a version of CRLB, we target to construct an estimator which saturates CRLB. For this purpose, we exploit the score vector in \eqref{eq:defScoreFuncDef}.
    \item[(ii)] Second, by exploiting the decomposition $\vec{\Delta}_w=\mathcal{H}_1 \oplus \mathcal{H}_2 \oplus \mathcal{H}_3$ introduced in Section~\ref{subsec:sepFISubspaces}, we construct an estimator based on the relationship between the projections of the unknown distribution $P$ and the score vector onto each of $\mathcal{H}_i$.
\end{enumerate}
Based on these ideas, we observe that by the symmetries of block designs, for any given $\alpha \in (0,1)$, the expectation of the projection of the score vector of a uBD mechanism at $P^{(\alpha)}$ (defined in \eqref{eq:DefPAlpha}) is a linear function of the projection of $P-P^{(\alpha)}$, where the linear coefficient is related with $M_i$ apperaing in Theorem \ref{thm:ULDP_opt_PUT} and $d_i:=\dim(\mathcal{H}_i)$.
\begin{proposition}\label{prop:scoreEstLinear}
    Let $Q$ be a $(w,v,\epsilon,t)$-uBD mechanism, and $\alpha \in (0,1)$. Then, for each $i=1,2,3$, we have
    \begin{align}
         \mathbb{E}_{Y \sim Q_P}[\Pi_{i}(\eta_{P^{(\alpha)},Q}(Y))] &=  \Pi_{i}(P - P^{(\alpha)}) \frac{d_i}{M_i(\alpha,t)}.\label{eq:scoreEstLinear}
    \end{align}
    Recall that $\Pi_i:\mathbb{R}^w \rightarrow \mathcal{H}_i$ is the orthogonal projection onto $\mathcal{H}_i$ defined in \eqref{eq:TanSpaceSubspace1}-\eqref{eq:TanSpaceSubspace3}, and $(d_1,d_2,d_3)=(v-1,w-v-1,1)$.
\end{proposition}
The proof of Proposition \ref{prop:scoreEstLinear} is in Appendix~\ref{app:proofScoreEstLinear}.
Using this, we can obtain an unbiased estimator for each of $\alpha$ by inverting the linear relationship for each $i=1,2,3$ to estimate each $\Pi_{i}(P - P^{(\alpha)})$, and aggregating the results for $i$'s. 
An exception is the case of $v \geq 2$, $t = \delta^{(v;v)}$, where the linear relationship is non-invertible due to $M_1(\alpha,t)=\infty$.
We note that the idea of score-based linear estimator was used in \cite{namOptimalPrivateDiscrete2024} for private estimation under LDP and a 1-bit communication constraint.
\begin{definition}{\label{def:uBDEstimator}}
    Let $Q$ be a $(w,v,\epsilon,t)$-uBD mechanism. Suppose that $t \neq \delta^{(v;v)}$ when $v \geq 2$. For each $\alpha \in [0,1]$, we define a sequence of unbiased estimators $\hat{P}^{(\alpha)}$ corresponding to $Q$, by $\hat{P}^{(\alpha)}_n (y^n) = \frac{1}{n}\sum_{j=1}^{n}\hat{P}^{(\alpha)}_1 (y_j)$, where
    \begin{enumerate}
        \item[(i)] For $\alpha \in (0,1)$, $\hat{P}^{(\alpha)}_1$ is given by\footnote{
        Here, if $d_i=0$, we regard $\frac{M_i(\alpha,t)}{d_i}\Pi_i(\eta_{P^{(\alpha)},Q}(y))=0$.
        }
            \begin{align}
                \hspace*{-5ex} \hat{P}^{(\alpha)}_1 (y) = P^{(\alpha)} + \sum_{i=1}^{3}\frac{M_i(\alpha,t)}{d_i}\Pi_i(\eta_{P^{(\alpha)},Q}(y)). \label{eq:uBDEstimator}
            \end{align}
        \item[(ii)] For $\alpha \in \{0,1\}$, $\hat{P}^{(\alpha)}_1$ is given by the continuous extension of \eqref{eq:uBDEstimator}.
    \end{enumerate}
    We call such a pair $(Q,\hat{P}^{(\alpha)})$ a \textbf{$(w,v,\epsilon,\alpha,t)$-uBD scheme}.
\end{definition}
To facilitate its implementation, we provide the table of each component of $\Pi_i(\eta_{P^{(\alpha)},Q}(y))$ for a uBD mechanism (or for a general extremal ULDP mechanism) in Table~\ref{tab:compScoreFuncAndProj} in the beginning of the appendix.

\subsection{Estimation Error and Optimality of uBD Schemes}
The estimation error of a uBD scheme can be represented using the objective function $M(\alpha,t)$ in Theorem~\ref{thm:ULDP_opt_PUT} and its derivative with respect to $\alpha$. Note that similar to block design schemes \cite{park2023block}, given $(w,v,\epsilon)$, the estimation error of a uBD scheme only depends on $(\alpha,t)$, and does not depend on specific choices of block designs $(\mathcal{X}_{\mathrm{S}}, \mathcal{E}_k)$.
\begin{proposition}\label{prop:uBDEstError}
    Let $(Q,\hat{P})$ be a $(w,v,\epsilon,\alpha,t)$-uBD scheme. Suppose that $t \neq \delta^{(v;v)}$ when $v \geq 2$. Then
    \begin{align}
        R_{n}(Q,\hat{P}) = \frac{1}{n}R_1(Q,\hat{P}), \quad R(Q,\hat{P}) = R_1(Q,\hat{P}), \label{eq:CanonicalEstMSE}
    \end{align}
    and
    \begin{multline}
        R(Q,\hat{P}) =
        \sup_{\beta \in [0,1]} \biggl[-\frac{w}{v(w-v)}(\beta-\alpha)^2 \\
        \hspace*{5em}+ (\beta-\alpha) F(\alpha,t) + M(\alpha,t)\biggr], \label{eq:uBDLoss}
    \end{multline}
    where $F(\alpha,t) = \frac{\partial}{\partial \alpha}M(\alpha,t)$.
\end{proposition}
\begin{proof}
    Since the estimator is unbiased and $Y_1,\cdots,Y_n$ are i.i.d. given $P$, \eqref{eq:CanonicalEstMSE} is straightforward. 
    Next, by the symmetry of a uBD scheme, the similar argument as \cite{park2023block} shows that whenever $P(\mathcal{X}_{\mathrm{S}})=\beta$, we have
    \begin{align}
        R_1(Q,\hat{P};P) \leq R_1(Q,\hat{P};P^{(\beta)}). \label{eq:uBDworstInPAlpha}
    \end{align}
    The detailed proof of \eqref{eq:uBDworstInPAlpha} is in Appendix~\ref{appsubsec:uBDworstInPAlpha}.
    Hence, we have
    \begin{align}
        R(Q,\hat{P}) = R_1(Q,\hat{P}) = \sup_{\beta \in [0,1]} R_1(Q,\hat{P};P^{(\beta)}). \label{eq:uBDworstInAlphaLimit}
    \end{align}
    In Appendix~\ref{appsubsec:uBDMSEAtPAlpha}, we show that
    \begin{multline}
        R_1(Q,\hat{P};P^{(\beta)})\\
        =-\frac{w}{v(w-v)}(\beta-\alpha)^2 + (\beta-\alpha) F(\alpha,t) + M(\alpha,t),\label{eq:uBDMSEAtPAlpha}
    \end{multline}
    which proves \eqref{eq:uBDLoss}.
\end{proof}
From this, we derive the asymptotic optimality of an optimized uBD scheme and finish the proof of Theorem \ref{thm:ULDP_opt_PUT}.
\begin{corollary}\label{cor:uBDOpt}
    Let $(\alpha^*,t^* )$ be a saddle point of $M(\alpha,t)$, and let $(Q,\hat{P})$ be a $(w,v,\epsilon,\alpha^*,t^*)$-uBD scheme. Then, $(Q,\hat{P})$ is asymptotically optimal, and $M^*(w,v,\epsilon)=M(\alpha^*, t^*)$.
\end{corollary}
\begin{proof}
    From the concavity of $M(\alpha,t)$ on $\alpha$, we have
    \begin{align}
        F(\alpha^*,t^* ) \begin{cases}
            =0 & (\text{if } \alpha^* \in (0,1))\\
            \leq 0 & (\text{if } \alpha^* = 0)\\
            \geq 0 & (\text{if } \alpha^* = 1)
        \end{cases}.
    \end{align}
    From this, it is clear that when $(\alpha,t)=(\alpha^*, t^* )$, the supremum in \eqref{eq:uBDLoss} is attained at $\beta=\alpha^* $, and thus $R(Q,\hat{P}) = M(\alpha^*,t^* )$. This implies that $M(\alpha^*,t^* ) \geq M^*(w,v,\epsilon)$. Since we have proved in Section~\ref{sec:converse} that $M(\alpha^*,t^* ) \leq M^*(w,v,\epsilon)$, we conclude that $M(\alpha^*,t^* ) = M^*(w,v,\epsilon)$, and thus $R(Q,\hat{P}) = M^*(w,v,\epsilon)$.
\end{proof}
\begin{remark}
    Corollary~\ref{cor:uBDOpt} means that an optimized uBD scheme achieves the minimum possible \emph{asymptotic} error among \emph{all} $(v,\epsilon)$-ULDP schemes, which are not necessarily unbiased.
    Moreover, by the same argument as in \cite{park2023block}, we can also show that for each of \emph{fixed $n$}, an optimized uBD scheme attains the minimum possible estimation error among all \emph{unbiased} $(v,\epsilon)$-ULDP schemes.
\end{remark}

\subsection{Closed-form Solution: Simple uBD Schemes}
Note that in all regimes where we find the closed-form solution as in Theorem~\ref{thm:closedFormExpPreciseForm}, $t^*$ is a point mass. There is a special behavior of uBD schemes such that $t$ is a point mass, which not only derives a simpler description of such schemes, but also plays an important role in the proof of Theorem~\ref{thm:closedFormExpPreciseForm}.
In short, the estimator $\hat{P}^{(\alpha)}$ does not depend on $\alpha$. 
\begin{proposition}\label{prop:simpleuBDEstInvarOverAlpha}
    Let $k \in [v]$. Let $Q$ be a $(w,v,\epsilon,\delta^{(k;v)})$-uBD mechanism. Assume that $k \neq v$ when $v \geq 2$. Then, the estimators $\hat{P}^{(\alpha)}$ in Definition~\ref{def:uBDEstimator} satisfy $\hat{P}^{(\alpha)}=\hat{P}^{(\alpha')}$ for every $\alpha,\alpha' \in [0,1]$. More precisely, the following holds for all $\alpha \in [0,1]$:
        \begin{enumerate}
        \item[(i)] For $x \in \mathcal{X}_{\mathrm{S}} = [v]$,
    \begin{multline}
        \!\!\!\!\!\!\!\!\!\!\!\!(\hat{P}^{(\alpha)}_1(y))_x =\\
        \begin{cases}
            1+\frac{v-1}{k(e^\epsilon-1)} & (\text{if }x \in y)\\
            -\frac{(k-1)(e^\epsilon-1)+(v-1)}{(v-k)(e^\epsilon-1)} & (\text{if }x \notin y, y \in \mathcal{Y}_{\mathrm{P}})\\
            -\frac{1}{k(e^\epsilon-1)} & (\text{if }y \in \mathcal{Y}_{\mathrm{I}})
        \end{cases}.\label{eq:simpleuBDEstSensitive}
    \end{multline}
        \item[(ii)] For $x \in \mathcal{X}_{\mathrm{N}} = [v+1:w]$,
    \begin{equation}
        (\hat{P}^{(\alpha)}_1(y))_x = \frac{(k(e^\epsilon-1)+v)}{k(e^\epsilon-1)} \mathbbm{1}(x \in y).\label{eq:simpleuBDEstNonSensitive}
    \end{equation}
    \end{enumerate}
\end{proposition}
Proposition~\ref{prop:simpleuBDEstInvarOverAlpha} can be shown by a direct calculation of \eqref{eq:uBDEstimator} to verify that it is equal to the above expression, hence we omit the proof.
Due to this proposition, we say a \textbf{$(w,v,\epsilon,k)$-simple uBD scheme} to mean a $(w,v,\epsilon,\alpha,\delta^{(k;v)})$-scheme, without ambiguity on specifying $\alpha$.
Proposition~\ref{prop:simpleuBDEstInvarOverAlpha} can be used to implement the estimator for a simple uBD scheme in a simpler way than using \eqref{eq:uBDEstimator}.
As a remark, the $(w,v,\epsilon,1)$-simple uBD scheme is equivalent to uRR \cite{murakami2019utility}.
% Also, the $(w,v,\epsilon,k)$-uBD scheme in which $\mathcal{E}_k$ is the complete block design can be seen as a variant of SS \cite{ye2018optimal}. However, this uBD scheme is not equivalent to uSS proposed in \cite{he2025addressing}.

Granting Theorem~\ref{thm:closedFormExpPreciseForm} about the closed-form solution in some regimes, which is proved in Appendix~\ref{app:proofClosedFormExpPreciseForm}, we deduce a closed-form characterization of optimal uBD schemes in the corresponding regimes. Interestingly, for Case (a), we deduce the optimality of uRR \cite{murakami2019utility}, which has not been shown before.
\begin{corollary}\label{cor:uBDk2optimal}
\,

    \begin{enumerate}
        \item[(a)] 
        Suppose that at least one of the following three conditions holds:
        \begin{enumerate}
            \item[(i)] $v=1$;
            \item[(ii)] $v \geq 2$ and $\epsilon \geq \ln\left(w-v+\sqrt{\dfrac{(w-1)(w-2)}{2}}\right)$;
            \item[(iii)] $v=2$ and $\epsilon \leq \ln\left(1+\sqrt{\dfrac{2(w-2)}{w-1}}\right)$;
        \end{enumerate}
        Then, uRR \cite{murakami2019utility} (equivalently, the $(w,v,\epsilon,1)$-simple uBD scheme) is asymptotically optimal.
        \item[(b)] Suppose that $v \geq 4$ and $\epsilon \leq \ln\sqrt{\tfrac{(v-1)(v-2)}{2}}$. Then, for each $k^* \in K^*(v,\epsilon) \cap [2:v-1]$, any $(w,v,\epsilon,k^*)$-simple uBD schemes are asymptotically optimal.
    \end{enumerate} 
\end{corollary}
\begin{remark}
    If $(\mathcal{X}_{\mathrm{S}},\mathcal{E}_k)$ is the complete block design in Remark~\ref{rmk:CBD}, the corresponding $(w,v,\epsilon,k)$-simple uBD scheme can be seen as a modification of the SS scheme \cite{ye2018optimal}. 
    However, this uBD scheme is not equivalent to the uSS scheme \cite{he2025addressing}, especially because the uBD scheme matches precisely one invertible element to each non-sensitive element, whereas the uSS scheme matches more than one invertible element, thereby introducing additional variance in the estimation.
\end{remark}

\begin{remark}\label{cor:comm_eff_uBD}
    The main advantage of considering the class of uBD schemes is that we can find an optimal and ``communication-efficient" ULDP scheme, similar to the class of block design schemes \cite{park2023block}.
    Here, the communication cost is defined as $\log_2|\supp(Q)|$ bits per client.
    Specifically, a straightforward generalization of the argument in \cite[Section IV-B]{park2023block} shows that in the regime of Case (b) in Theorem~\ref{thm:closedFormExpPreciseForm}, if there exists a block design with $b=v$ and $k \in K^*(v,\epsilon) \cap [2:v-1]$, then there exists a simple uBD scheme which is not only optimal, but also has communication cost of $\log_2 w$.
    This amount of communication is the smallest possible communication cost required for a consistent estimation \cite{park2023block, acharya2019communication}. Note that in the regime of Case (a), uRR \cite{murakami2019utility} achieves optimal PUT and the communication cost of $\log_2 w$.
\end{remark}

The region where a closed-form solution cannot be obtained through Theorem~\ref{thm:closedFormExpPreciseForm} is $v \geq 2$ and $\epsilon \in (\epsilon_L,\epsilon_H)$, where
\begin{align}
    \epsilon_L &= \begin{cases}
        \ln \left(1+\sqrt{\frac{2(w-2)}{w-1}} \right) & (\text{if }v=2)\\
        \ln \sqrt{\frac{(v-1)(v-2)}{2}} & (\text{if }v \geq 3)
    \end{cases},\label{eq:epsLowDef}\\
    \epsilon_H &=  \ln\left(w-v+\sqrt{\frac{(w-1)(w-2)}{2}}\right).\label{eq:epsHighDef}
\end{align}
In the next section, while performing experiments on a real dataset, we find the optimal parameters $(\alpha,t)$ by \emph{MATLAB Optimization Toolbox} in the region above and analyze the result.

\section{Experiments}\label{sec:experiment}
\begin{figure*}[h!]
    \centering
    \begin{subfigure}{0.8\linewidth}
        \makebox[\linewidth][c]{\includegraphics[width = 0.5\linewidth]{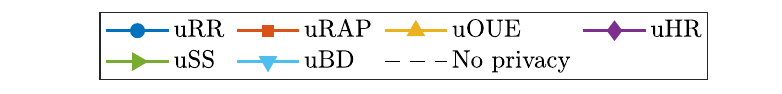}}
    \end{subfigure}
        \begin{subfigure}{0.32\linewidth}
        \makebox[\linewidth][c]{\includegraphics[width=\linewidth]{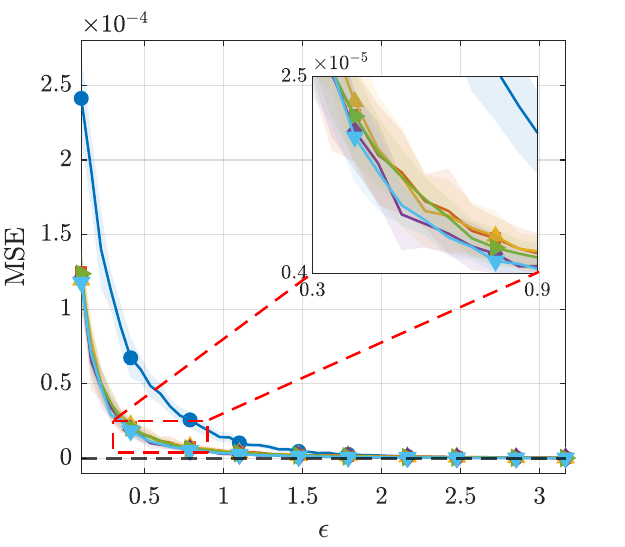}}
        \caption{$\epsilon\in\left[0.1,\, \epsilon_L\right]$}
    \end{subfigure}
    \begin{subfigure}{0.32\linewidth}
        \makebox[\linewidth][c]{\includegraphics[width=\linewidth]{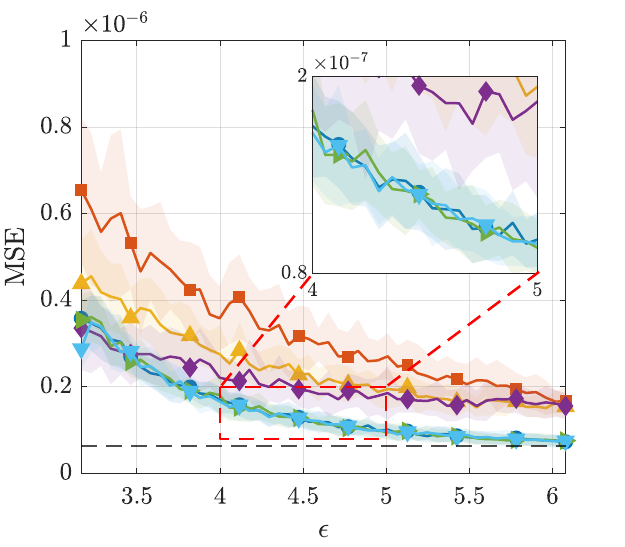}}
        \caption{$\epsilon\in\left[\epsilon_L,\,\epsilon_H \right]$}
    \end{subfigure}
    \begin{subfigure}{0.32\linewidth}
        \makebox[\linewidth][c]{\includegraphics[width=\linewidth]{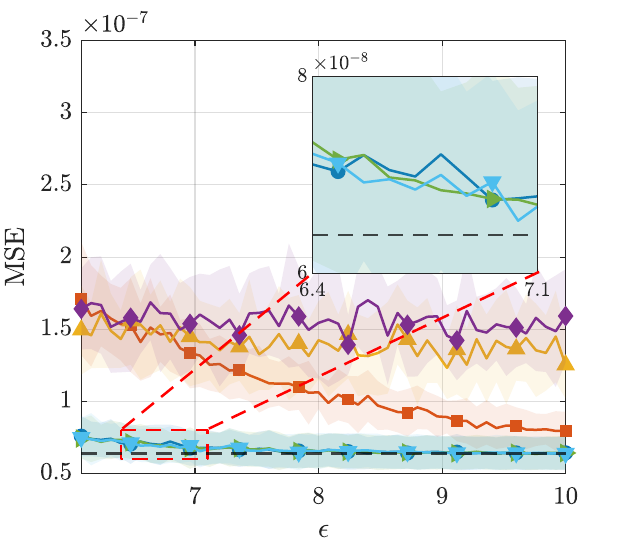}}
        \caption{$\epsilon\in[\epsilon_H,\,10]$}
    \end{subfigure}
    \caption{$\epsilon$ vs. PUT where $(w,v) = (277,35)$.}
    \label{fig:exp_(277,35)}
\end{figure*}
\begin{figure*}[h!]
    \centering
    \begin{subfigure}{0.32\linewidth}
        \makebox[\linewidth][c]{\includegraphics[width=\linewidth]{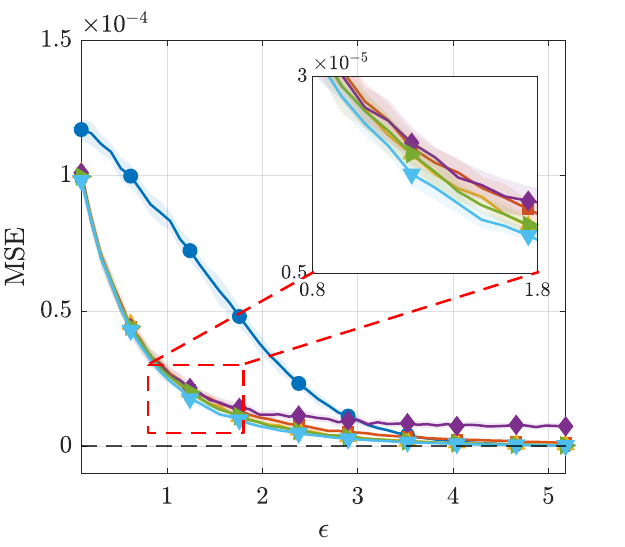}}
        \caption{$\epsilon\in\left[0.1,\, \epsilon_L\right]$}
    \end{subfigure}
    \begin{subfigure}{0.32\linewidth}
        \makebox[\linewidth][c]{\includegraphics[width=\linewidth]{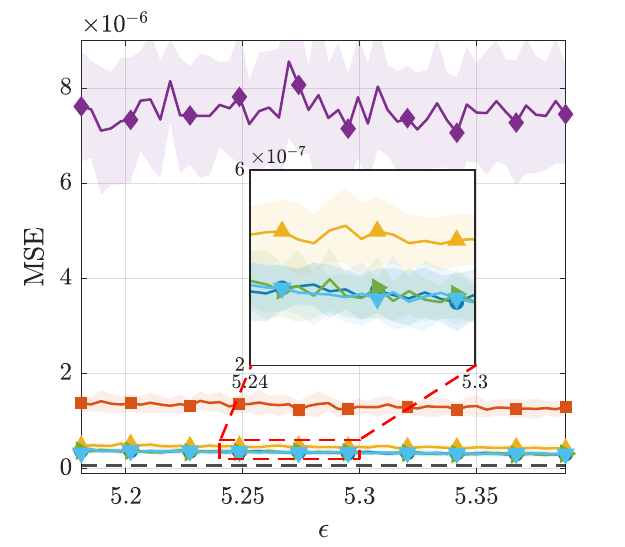}}
        \caption{$\epsilon\in\left[\epsilon_L,\,\epsilon_H \right]$}
    \end{subfigure}
    \begin{subfigure}{0.32\linewidth}
        \makebox[\linewidth][c]{\includegraphics[width=\linewidth]{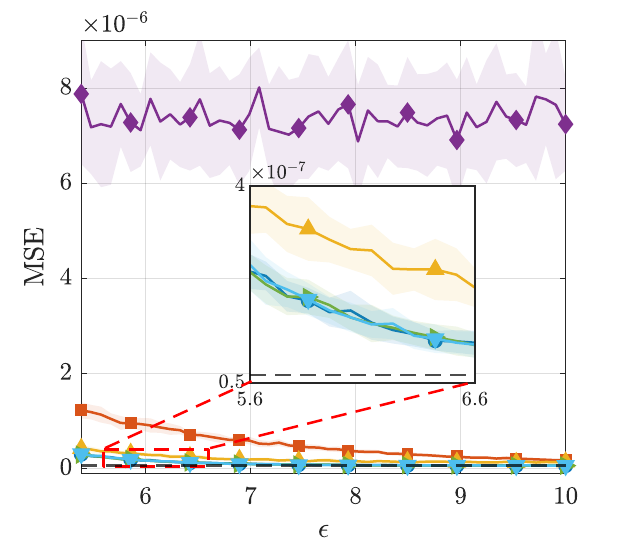}}
        \caption{$\epsilon\in[\epsilon_H,\,10]$}
    \end{subfigure}
    \caption{$\epsilon$ vs. PUT where $(w,v) = (277,253)$.}
    \label{fig:exp_(277,253)}
\end{figure*}
\begin{figure*}[h!]
    \centering
    \begin{subfigure}{0.32\linewidth}
        \makebox[\linewidth][c]{\includegraphics[width =\linewidth]{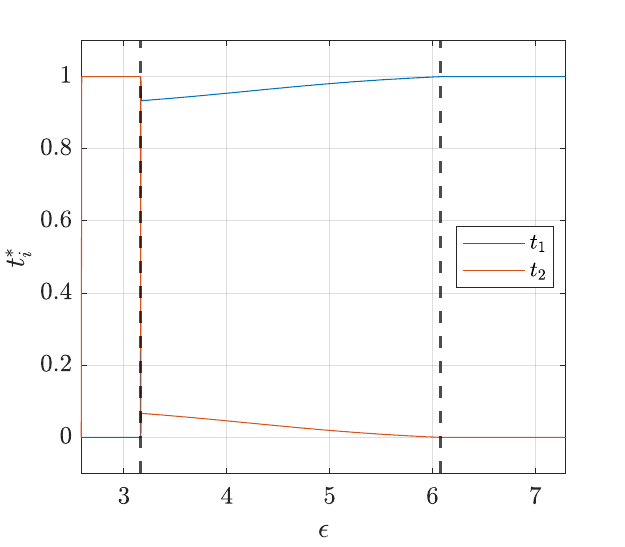}}
        \caption{$(w,v)=(277,35)$}
    \end{subfigure}
    \begin{subfigure}{0.32\linewidth}
        \makebox[\linewidth][c]{\includegraphics[width =\linewidth]{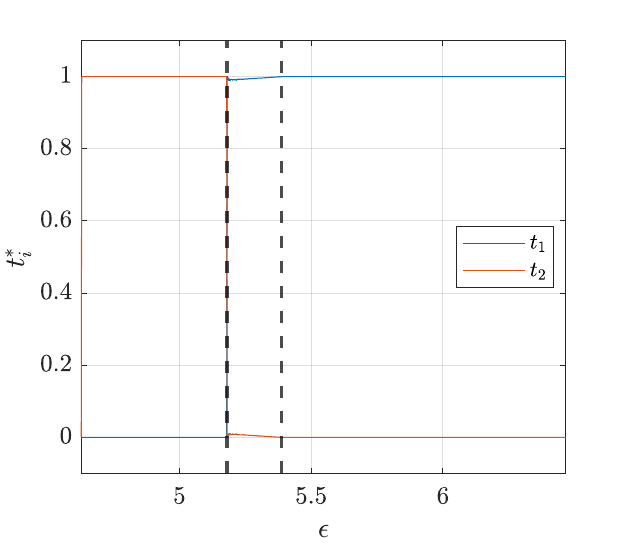}}
        \caption{$(w,v)=(277,253)$}
    \end{subfigure}
    \caption{$\epsilon$ vs. $t^*_i$, where the vertical lines represent $\epsilon_L$ and $\epsilon_H$.}
    \label{fig:exp_t_i}
\end{figure*}
We empirically evaluate the PUTs of uBD and baseline schemes \cite{murakami2019utility, liu2024multidimensional, acharya2020context, he2025addressing}.
We use the American Community Survey (ACS) 2023 1-year public use microdata sample (PUMS) \cite{acs_pums_2023_1y}, which contains records for $n = 1,732,343$ individuals.
Six attributes are retained: age, educational attainment, marital status, disability record, employment status, and income-to-poverty ratio, whose category counts are $(4,2,2,3,2,3)$, giving $288$ in total; excluding empty nominal categories yields $w = 277$ unique input alphabets.

We applied two alternative criteria:
\begin{enumerate}
    \item[(i)] \textbf{Stringent criterion} ${(v = 35)}$: A value is sensitive only if all three conditions hold simultaneously:
    \begin{itemize}
        \item educationally disadvantaged, widowed/divorced, or disabled,
        \item absence from the labor force,
        \item income below the poverty threshold.
    \end{itemize}
    \item[(ii)] \textbf{Permissive criterion} ${(v = 253)}$: A value is sensitive if it satisfies any one of the above conditions.
\end{enumerate}
The stringent rule labels $79{,}141$ records as sensitive, while the permissive rule identifies $1{,}022{,}617$.

For each criterion, we randomly sample $50,000$ users twenty times and report the average PUT across these subsamples. 
Guided by Theorem~\ref{thm:closedFormExpPreciseForm}, we partition the privacy budget into three intervals: $[0.1, \epsilon_L]$, $[\epsilon_L, \epsilon_H]$, and $[\epsilon_H, 10]$, where 
$(\epsilon_L,\epsilon_H)$ are in \eqref{eq:epsLowDef}-\eqref{eq:epsHighDef}.
Figs. \ref{fig:exp_(277,35)} and \ref{fig:exp_(277,253)} show that our scheme consistently outperforms all baselines, regardless of privacy regimes and sensitive-set sizes.
In the intermediate regime $\epsilon\in(\epsilon_L, \epsilon_H)$, where a closed-form expression of the optimal parameter $(\alpha^*,t^*)$ is unknown, we plot the empirical values of $t^*$ in Fig.~\ref{fig:exp_t_i}.
Remarkably, $t^*$ collapses to the two-point form $t^* = (t_1,t_2,0,\dots,0)$.\footnote{Components below $10^{-8}$ are treated as numerical noise and forced to zero; after which the vector is renormalized.}
This observation leads us to carefully conjecture that utilizing a single block design is not sufficient, but utilizing just two block designs—those with $k=1$ and $k=2$—is sufficient to attain the optimal PUT throughout $\epsilon\in(\epsilon_L,\epsilon_H)$.
The MATLAB codes for experiments are available at \url{https://github.com/phy811/uBD}.

\section{Conclusion}\label{sec:concl}
In this paper, we characterized the optimal PUT for discrete distribution estimation under the ULDP constraint. We proved the converse by leveraging a uniform asymptotic Cram\'er-Rao lower bound, a reduction to extremal ULDP mechanisms, and a distribution decomposition argument. Also, we proposed a class of PUT-optimal uBD schemes, obtained through nontrivial modifications of block design mechanisms and by employing a score-based estimator together with a distribution decomposition argument.

An interesting direction for future research is to derive closed-form solutions to our optimization problem in the intermediate privacy regime, thereby specifying the optimal parameters for the estimator $\alpha^*$ and the mixture proportion $t^*$.
Also, we could consider other formulations of PUT, such as instance-optimality setup \cite{steinbergerEfficiencyLocalDifferential2024} rather than worst-case setup, and considering general $\ell_p$ loss or $f$-divergence loss (such as relative entropy) other than MSE loss ($\ell_2$).
It is also of interest to extend the analysis to a broader class of privacy metrics beyond LDP and ULDP.
We believe that the analytical tools developed in this paper--particularly the uniform asymptotic Cram\'er-Rao lower bound--will provide a valuable foundation for pursuing these directions.

\bibliographystyle{ieeetr}
\bibliography{IEEEabrv, ref}

\newpage
\appendices

\begingroup
\renewcommand{\arraystretch}{2}
\newcommand{\balstrut}{\rule[-3ex]{0pt}{7ex}}
\newcolumntype{C}{>{\centering\arraybackslash\balstrut}c}
\begin{table*}[htbp]
    \centering
    \begin{NiceTabular}{|C||CCC|CCC|}
\hline
\multirow{2}{*}{} & \multicolumn{3}{c}{$\displaystyle x \in \mathcal{X}_{\mathrm{S}}$}                                 & \multicolumn{3}{c}{$\displaystyle x \in \mathcal{X}_{\mathrm{N}}$}                                 \\ \cline{2-7} 
                  & \multicolumn{1}{c}{{$\displaystyle y \in \mathcal{Y}_{\mathrm{P}}^{(k)}$, $\displaystyle x \in y$}}\vline & \multicolumn{1}{c}{{$\displaystyle y \in \mathcal{Y}_{\mathrm{P}}^{(k)}$, $\displaystyle x \notin y$}}\vline & $\displaystyle y \in \mathcal{Y}_{\mathrm{I}}$ & \multicolumn{1}{c}{$\displaystyle y \in \mathcal{Y}_{\mathrm{P}}^{(k)}$}\vline & \multicolumn{1}{c}{$\displaystyle y \in \mathcal{Y}_{\mathrm{I}}, x \in y$}\vline & $\displaystyle y \in \mathcal{Y}_{\mathrm{I}}, x \notin y$ \\ \hline
 $\displaystyle(\eta_{P^{(\alpha),Q}}(y))_x$                 & \multicolumn{1}{c}{$\displaystyle\frac{ve^\epsilon}{\alpha k(e^\epsilon-1)+v}$}   & \multicolumn{1}{c}{$\displaystyle\frac{v}{\alpha k(e^\epsilon-1)+v}$}   &  $\displaystyle0$  & \multicolumn{1}{c}{$\displaystyle\frac{v}{\alpha k(e^\epsilon-1)+v}$}   & \multicolumn{1}{c}{$\displaystyle\frac{w-v}{1-\alpha}$}   &  $0$  \\ \hline
$\displaystyle\left(\Pi_1\left(\eta_{P^{(\alpha),Q}}(y)\right)\right)_x$                & \multicolumn{1}{c}{$\displaystyle\frac{(v-k)(e^\epsilon-1)}{\alpha k(e^\epsilon-1)+v}$}   & \multicolumn{1}{c}{$\displaystyle\frac{-k(e^\epsilon-1)}{\alpha k(e^\epsilon-1)+v}$}   &  $\displaystyle0$  & \multicolumn{1}{c}{$\displaystyle0$}   & \multicolumn{1}{c}{$\displaystyle0$}   &  $\displaystyle0$  \\ \hline
$\displaystyle\left(\Pi_2\left(\eta_{P^{(\alpha),Q}}(y)\right)\right)_x$                 & \multicolumn{1}{c}{$\displaystyle0$}   & \multicolumn{1}{c}{$\displaystyle0$}   &  $\displaystyle0$  & \multicolumn{1}{c}{$\displaystyle0$}   & \multicolumn{1}{c}{$\displaystyle\frac{w-v-1}{1-\alpha}$}   &  $\displaystyle-\frac{1}{1-\alpha}$  \\ \hline
$\displaystyle\left(\Pi_3\left(\eta_{P^{(\alpha),Q}}(y)\right)\right)_x$                   & \multicolumn{1}{c}{$\displaystyle\frac{k(w-v)(e^\epsilon-1)}{w(\alpha k(e^\epsilon-1)+v)}$}   & \multicolumn{1}{c}{$\displaystyle\frac{k(w-v)(e^\epsilon-1)}{w(\alpha k(e^\epsilon-1)+v)}$}   &  $\displaystyle-\frac{w-v}{w(1-\alpha)}$  & \multicolumn{1}{c}{$\displaystyle\frac{-vk(e^\epsilon-1)}{w(\alpha k(e^\epsilon-1)+v)}$}   & \multicolumn{1}{c}{$\displaystyle\frac{v}{w(1-\alpha)}$}   &   $\displaystyle\frac{v}{w(1-\alpha)}$ \\ \hline
\end{NiceTabular}
    \caption{Components of Score Vectors and Their Projections for Extremal ULDP Mechanisms.
    $\mathcal{Y}_{\mathrm{P}}^{(k)}$ is defined in Appendix~\ref{app:additionalNotation}.
    }
    \label{tab:compScoreFuncAndProj}
\end{table*}
\endgroup

\section{Additional Notations}\label{app:additionalNotation}
\newcommand{\VAR}{\mathrm{VAR}}
For an $\mathbb{R}^n$-valued random variable $Z$, we denote $\VAR(Z):=\sum_{i=1}^{n}\Var(Z_i)$.

Suppose that $v$ is given. For each $k\in[v]$, we define $\mathcal{Y}_{\mathrm{P}}^{(k)}:=\{y \subset [v]:|y|=k\}$.

\section{Convexity and Attainment of the Optimization Problem}\label{app:objCvxandAttainSaddle}
In this appendix, we briefly explain the reason that the objective function $M(\alpha,t)$ has a saddle point and is concave-convex.

First, we show that $M$ is concave-convex.
For a fixed $\alpha$, each of $M_i(\alpha,t)$ is convex in $t$, being a composition of a linear function with a convex function $z \mapsto 1/z, z \in [0,\infty)$. Also, for a fixed $t$, each of $M_i(\alpha,t)$ is proportional to the weighted harmonic mean of concave functions on $\alpha$ over $k$ with weights $t_k$, where the corresponding concave functions for $i=1,2,3$ are $\tfrac{(\alpha k(e^\epsilon-1)+v)(ke^\epsilon+v-k)}{k(v-k)}$, $\tfrac{(1-\alpha)(ke^\epsilon+v-k)}{k}$, and $\tfrac{(1-\alpha)(\alpha k(e^\epsilon-1)+v)}{k}$, respectively. It is a standard fact \cite{boyd_convex_2004} that the harmonic mean is concave; thus, each of $M_i(\alpha,t)$ is concave in $\alpha$.
Thus, we conclude that the objective function $M(\alpha,t)$ is concave-convex in $(\alpha,t)$.

Next, we show that $M$ has a saddle point. For $v=1$, we have $\mathrm{dom}(M) = [0,1] \times \Delta_v$, and for $v \geq 2$, we have $\mathrm{dom}(M) = [0,1] \times (\Delta_v \backslash \{\delta^{(v;v)}\})$. Also, $M$ is continuous everywhere in $[0,1] \times \Delta_v$. Since $[0,1]$ is compact, by Sion's minimax theorem \cite{sionGeneralMinimaxTheorems1958}, the saddle point of $M$ exists.

\section{Detailed Proof for Section~\ref{sec:converse}}\label{app:DetailsForConverse}
\subsection{Statement and Proof of Equicontinuity Result in Theorem~\ref{thm:unifAsympCRLB}}\label{appsubsec:proofPropFIBwTwoDistIneq}
\begin{proposition}\label{prop:FIBwTwoDistIneq}
        Let $P,P' \in \Delta_w^\circ$. Then, for any conditional distribution $Q$ from $[w]$ to a finite set $\mathcal{Y}$, we have
        \begin{align}
         \mathcal{J}_{P',Q} \preceq \left(\max_{x \in [w]}\frac{P_x}{P'_x} \right) \mathcal{J}_{P, Q}. \label{eq:FIBwTwoDistInequ}
        \end{align}
\end{proposition}
\begin{proof}
Let $r=\max_{x \in [w]}\frac{P_x}{P'_x}$. We need to show that: for any $h \in \vec{\Delta}_w$, we have
\begin{align}
    \mathcal{J}_{P',Q}(h,h) \leq r \mathcal{J}_{P,Q}(h,h).
\end{align}
Observe that
\begin{align}
    {Q_{P'}(y)} &= \sum_{x \in [w]}P'_xQ(y|x) \\ &\geq \sum_{x \in [w]}({P_x}/{r})Q(y|x) = {Q_P(y)}/{r}.
\end{align}
Thus, we obtain
    \begin{align}
        \mathcal{J}_{P',Q}(h,h) &= \mathrm{Var}_{Y \sim Q_{P'}}(\eta_{P',Q;h}(Y)) \\
        &= \sum_{y \in \supp(Q)} \frac{\left(\sum_{x \in [w]} h_x Q(y|x) \right)^2}{Q_{P'}(y)}\\
        &\leq r\sum_{y \in \supp(Q)} \frac{\left(\sum_{x \in [w]} h_x Q(y|x) \right)^2}{Q_{P}(y)}\\
        &= r \mathcal{J}_{P,Q}(h,h).
    \end{align}
\end{proof}

\subsection{Derivation of (\ref{eq:FIExpectationUnifBound}) in Theorem~\ref{thm:unifAsympCRLB}}\label{appsubsec:proofFIExpectationUnifBound}
First, observe that
\begin{align}
    \frac{P_x}{\tilde{P}_x^{(\theta)}}=\frac{P_x}{P_x + \sum_{i=1}^{w-1}\theta_i h^{(i)}_x} = \frac{1}{1 + \sum_{i=1}^{w-1}\frac{\theta_i h^{(i)}_x}{P_x}}.
\end{align}
Since $\Vert h^{(i)}\Vert=1$, we have $\left|\sum_{i=1}^{w-1}\frac{\theta_i h^{(i)}_x}{P_x}\right| \leq \frac{(w-1) \max_i |\theta_i|}{P_x}$. Also, note that $\frac{1}{1-z}\leq 1+z$ for $z \in (0,1)$. Thus, letting
\begin{align}
    N&=\max\left(\frac{(w-1)^2}{D^4}, \left(\frac{(w-1)}{\min_x P_x}\right)^4 \right)+1,\\
    C&=\frac{(w-1)}{\min_x P_x},
\end{align}
we derive that whenever $n \geq N$ and $\Theta \sim \phi_n$, we have $\frac{(w-1) \max_i |\Theta_i|}{P_x} \leq C n^{-1/4} < 1$ almost surely, and thus
\begin{align}
    \max_{x}\frac{P_x}{\tilde{P}_x^{(\Theta)}} \leq \frac{1}{1-C n^{-1/4}}\leq 1+Cn^{-1/4}.
\end{align}
Combining this with Proposition~\ref{prop:FIBwTwoDistIneq}, we derive \eqref{eq:FIExpectationUnifBound}.

\subsection{Derivation of (\ref{eq:changeLimitAndInfCompact}) in Theorem~\ref{thm:unifAsympCRLB}}\label{appsubsec:proofPropChangeLimitAndInfCompact}
We first remark that $J_{P,Q}$ is continuous in $Q$ but it can be singular, and hence $\trInv\left(J_{P, Q}\right)$ is extended real-valued continuous function on $Q$.
To address this issue, we note that there is an order-preserving homeomorphism $\sigma:\mathbb{R} \cup \{-\infty,\infty\} \rightarrow [0,1]$, e.g., a logistic function $\sigma(x)=\frac{1}{1+e^{-x}}$.
Then, by Dini's theorem, $\sigma(\overline{R}_n(P,Q)) \uparrow \sigma \left(\trInv\left(J_{P, Q}\right)\right)$ \emph{uniformly} as $n \rightarrow \infty$, and thus
\begin{align}
    \lim_{n \rightarrow \infty} \inf_{Q \in \mathcal{Q}} \sigma(\overline{R}_n(P,Q)) = \inf_{Q \in \mathcal{Q}}\sigma \left(\trInv\left(J_{P, Q}\right)\right).
\end{align}
By taking $\sigma^{-1}$ on both sides, we obtain \eqref{eq:changeLimitAndInfCompact}.

\subsection{Proof of Theorem~\ref{thm:ULDPExtremal}}\label{app:proofULDPExtremal}
Let an $(v,\epsilon)$-ULDP mechanism $\tilde{Q}:\mathcal{X}\rightarrow \mathcal{P}(\tilde{\mathcal{Y}})$ be given, whose sets of protected and invertible data are $\tilde{\mathcal{Y}}_{\mathrm{P}}$ and $\tilde{\mathcal{Y}}_{\mathrm{I}}$, respectively. Also, let $\mathcal{Y}=\mathcal{Y}_{\mathrm{P}} \sqcup \mathcal{Y}_{\mathrm{I}}$, $\mathcal{Y}_{\mathrm{P}} = 2^{[v]} \backslash \{\emptyset\}$, and $\mathcal{Y}_{\mathrm{I}} = \{\{x\}: x \in [v+1:w]\}$.
We need to find $Q \in \mathcal{Q}^{\mathrm{E}}_{w,v,\epsilon}$ and $T:\mathcal{Y} \rightarrow \mathcal{P}(\tilde{\mathcal{Y}})$ such that
\begin{align}
    \tilde{Q}(\tilde{y}|x)=\sum_{y \in \mathcal{Y}}Q(y|x)T(\tilde{y}|y), \quad \forall x\in\mathcal{X},\tilde{y} \in \tilde{\mathcal{Y}}. \label{eq:degradeDef}
\end{align}
Suppose that $Q \in \mathcal{Q}_{w,v,\epsilon}^{\mathrm{E}}$ is induced from $\gamma:\mathcal{Y}_{\mathrm{P}} \rightarrow \mathbb{R}_{\geq 0}$ as in Definition~\ref{def:ULDPExtremal}. Then, \eqref{eq:degradeDef} is equivalent to that the following holds for every $\tilde{y} \in \tilde{\mathcal{Y}}$:
\begin{align}
    \tilde{Q}_{\tilde{y}} = &\sum_{y \in \mathcal{Y}_{\mathrm{P}}}T(\tilde{y}|y)\gamma(y) s^{(y)} \nonumber\\
    &+ \sum_{x \in \mathcal{X}_{\mathrm{N}}}T(\tilde{y}|\{x\})\left(1-\sum_{y' \in \mathcal{Y}_{\mathrm{P}}}\gamma(y')\right)\delta^{(x)},\label{eq:degradeCond}
\end{align}
where $\tilde{Q}_{\tilde{y}}:=(\tilde{Q}(\tilde{y}|1),\cdots,(\tilde{y}|w))$, $\delta^{(x)}:=\delta^{(x;w)}$, and for each $z \subset \mathcal{X}$, $s^{(z)} \in \{1,e^\epsilon\}^w$ is defined by
\begin{align}
    s_x^{(z)}=\begin{cases}
        e^\epsilon & (\text{if }x \in z)\\
        1 & (\text{if }x \notin z)
    \end{cases}.
\end{align}
Therefore, it suffices to find $\gamma$ and $T$ which satisfy \eqref{eq:degradeCond}.

By \cite[Lemma 12]{ye2018optimal} and the definition of protected data, for each $\tilde{y} \in \tilde{\mathcal{Y}}_{\mathrm{P}}$, $\tilde{Q}_{\tilde{y}}$ is in a conic hull of $\{s^{(z)}:z \subset \mathcal{X}\}$. Moreover, for each $z \subset \mathcal{X}$, $s^{(z)}$ is in a conic hull of $\mathfrak{S}:=\{s^{(y)}:y \in \mathcal{Y}_{\mathrm{P}}\}\cup\{\delta^{(x)}:x \in \mathcal{X}_{\mathrm{N}}\}$, as follows:
\begin{itemize}
    \item If $z \cap \mathcal{X}_{\mathrm{S}} \neq \emptyset$, then
    \begin{align}
        s^{(z)}=s^{(z \cap \mathcal{X}_{\mathrm{S}})} + (e^\epsilon-1)\sum_{x \in \mathcal{X}_{\mathrm{N}} \cap z} \delta^{(x)}.
    \end{align}
    \item If $z \cap \mathcal{X}_{\mathrm{S}} = \emptyset$, then
    \begin{align}
        s^{(z)}=e^{-\epsilon} s^{(\mathcal{X}_{\mathrm{S}})} &+ (e^\epsilon-e^{-\epsilon})\sum_{x \in \mathcal{X}_{\mathrm{N}} \cap z} \delta^{(x)} \nonumber\\
        &+ (1-e^{-\epsilon})\sum_{x \in \mathcal{X}_{\mathrm{N}} \backslash z} \delta^{(x)}.
    \end{align}
\end{itemize}
Thus, for each $\tilde{y} \in \tilde{\mathcal{Y}}_P$, $\tilde{Q}_{\tilde{y}}$ is in a conic hull of $\mathfrak{S}$. Moreover, for each $\tilde{y} \in \tilde{\mathcal{Y}}_{\mathrm{I}}$, $\tilde{Q}_{\tilde{y}}$ is a constant multiple of $\delta^{(x)}$ for some $x \in \mathcal{X}_{\mathrm{N}}$, which is also in a conic hull of $\mathfrak{S}$. Thus, for each $\tilde{y} \in \tilde{\mathcal{Y}}$, we can write
\begin{align}
    \tilde{Q}_{\tilde{y}} = \sum_{y \in \mathcal{Y}_{\mathrm{P}}} \gamma(y,\tilde{y}) s^{(y)} + \sum_{x \in \mathcal{X}_{\mathrm{N}}} u(x,\tilde{y}) \delta^{(x)} \label{eq:tildeQConicComb}
\end{align}
for some $\gamma(y,\tilde{y}), u(x,\tilde{y}) \geq 0$.
Using this, we can observe that \eqref{eq:degradeCond} is achieved by the following $\gamma$ and $T$:
\begin{align}
    \gamma(y)&=\sum_{\tilde{y} \in \tilde{\mathcal{Y}}} \gamma(y,\tilde{y}), \\
    T(\tilde{y}|y)&=
    \begin{cases}
        \gamma(y,\tilde{y})/\gamma(y) & (\text{if }y \in \mathcal{Y}_{\mathrm{P}})\\
        \frac{u(x,\tilde{y})}{1-\sum_{y' \in \mathcal{Y}_{\mathrm{P}}}\gamma(y')} & (\text{if }y \in \mathcal{Y}_{\mathrm{I}}, y=\{x\})
    \end{cases}.
\end{align}
Here, if $y \in \mathcal{Y}_{\mathrm{P}}$ satisfies $\gamma(y)=0$, then we set $T(\cdot|y)$ to be an arbitrary PMF.

From \eqref{eq:tildeQConicComb} and $\sum_{\tilde{y} \in \tilde{\mathcal{Y}}} \tilde{Q}_{\tilde{y}}=\mathbf{1}_w$, it is straightforward to see that above $\gamma$ satisfies \eqref{eq:ULDPGammaCond} to induce a valid extremal $(v,\epsilon)$-ULDP mechanism, and above $T$ is a valid conditional distribution satisfying $\sum_{\tilde{y} \in \tilde{\mathcal{Y}}}T(\tilde{y}|y)=1$. This finishes the proof.

\subsection{Derivation of (\ref{eq:CRLBDecompIntoSubspaces}) in Proposition~\ref{prop:ConvPart}}\label{app:proofCRLBDecompIntoSubspaces}
The block matrix inversion formula \cite{boyd_convex_2004} implies that for a positive definite matrix $A$ with $A=\begin{bmatrix}
    B & C^T \\ C & D
\end{bmatrix}$, where both $B$ and $D$ are square submatrices, $A^{-1}$ has the following block form
\begin{align}
    A^{-1} = \begin{bmatrix}
        (B-C^T D^{-1}C)^{-1} & \tilde{C}^T\\
        \tilde{C}^T & (D-CB^{-1}C^T)^{-1}
    \end{bmatrix},
\end{align}
where both $B-C^T D^{-1}C$ and $D-CB^{-1}C^T$ are positive definite by the Schur complement lemma \cite{boyd_convex_2004}.
This result, together with $A \preceq B \Rightarrow A^{-1} \succeq B^{-1}$ \cite{bikchentaevTraceInequalitiesMatrices2024}, we obtain
\begin{align}
    \trInv(A) \geq \trInv(B)+\trInv(D).\label{eq:invTrIneq}
\end{align}
\eqref{eq:invTrIneq} also holds when $A$ is non-invertible positive semidefinite, by the definition that $\trInv(A)=\infty$. Now, \eqref{eq:CRLBDecompIntoSubspaces} follows directly from applying \eqref{eq:invTrIneq} twice.

\subsection{Derivation of (\ref{eq:TraceFIRestPropToJi}) in Proposition~\ref{prop:ConvPart}}\label{app:proofConvPart}
For each $i=1,2,3$, we have 
\begin{align}
    \tr(B_i)=\sum_{j=1}^{d_i}\mathcal{J}_{P^{(\alpha),Q}}(h^{(j;i)},h^{(j;i)})
\end{align}
for an orthonormal basis $\{h^{(j;i)}\}_{j=1}^{d_i}$ of $\mathcal{H}_i$. By definition of Fisher information and the standard fact \cite{polyanskiyInformationTheoryCoding2024} that the score has zero mean, we have
\begin{align}
    \tr(B_i) &= \sum_{j=1}^{d_i}\Var_{Y \sim Q_{P^{(\alpha)}}}\left[\langle \eta_{P^{(\alpha)},Q}(Y),h^{(j;i)} \rangle\right]\\
    &= \sum_{j=1}^{d_i}\mathbb{E}_{Y \sim Q_{P^{(\alpha)}}}\left[\langle \eta_{P^{(\alpha)},Q}(Y),h^{(j;i)} \rangle^2\right]\\
    &= \mathbb{E}_{Y \sim Q_{P^{(\alpha)}}}\left\Vert\Pi_i\left(\eta_{P^{(\alpha)},Q}(Y)\right) \right\Vert^2. \label{eq:trFIandNormProj}
\end{align}
Now, suppose that $Q$ is induced from $\gamma:\mathcal{Y}_{\mathrm{P}} \rightarrow \mathbb{R}_{\geq 0}$ as in Definition~\ref{def:ULDPExtremal}. Also, for simplicity, let $t:=t(Q)$. It is straightforward to see that for each $k \in [v]$, we have
\begin{align}
    t_k = \sum_{y \in \mathcal{Y}_{\mathrm{P}}^{(k)}} \frac{ke^\epsilon+v-k}{v}\gamma(y).
\end{align}
Also, we have
\begin{align}
    Q_{P^{(\alpha)}}(\mathcal{Y}_{\mathrm{P}}^{(k)}) &= \sum_{y \in \mathcal{Y}_{\mathrm{P}}^{(k)}} \frac{\alpha k(e^\epsilon-1)+v}{v} \gamma(y)\\
    &= \frac{\alpha k(e^\epsilon-1)+v}{ke^\epsilon+v-k} t_k
\end{align}
and
\begin{align}
    Q_{P^{(\alpha)}}(\mathcal{Y}_{\mathrm{I}}) 
    &= 1-\sum_{k=1}^{v}  Q_{P^{(\alpha)}}(\mathcal{Y}_{\mathrm{P}}^{(k)})\\
    &= \sum_{k=1}^{v}\frac{(1-\alpha)(e^\epsilon-1)k}{ke^\epsilon+v-k}t_k.
\end{align}
Furthermore, referring Table~\ref{tab:compScoreFuncAndProj}, we obtain that
\begin{itemize}
    \item If $y \in \mathcal{Y}_{\mathrm{P}}^{(k)}$, then $\left\Vert\Pi_i\left(\eta_{P^{(\alpha)},Q}(y)\right) \right\Vert^2=A_{i,k}$, where
\end{itemize}
\begin{multline}
    (A_{1,k},A_{2,k},A_{3,k}) =\\
    \left(\frac{vk(v-k)(e^\epsilon-1)^2}{(\alpha k(e^\epsilon-1)+v)^2}, 0, \frac{k^2(e^\epsilon-1)^2 v(w-v)}{w(\alpha k(e^\epsilon-1)+v)^2}\right).\label{eq:normScoreProtected}
\end{multline}
\begin{itemize}
    \item If $y \in \mathcal{Y}_{\mathrm{I}}$, then $\left\Vert\Pi_i\left(\eta_{P^{(\alpha)},Q}(y)\right) \right\Vert^2=A_i'$, where
\end{itemize}
\begin{multline}
    (A_1',A_2',A_3') =\\
    \left(0,\frac{(w-v)(w-v-1)}{(1-\alpha)^2}, \frac{v(w-v)}{w(1-\alpha)^2}\right).\label{eq:normScoreInvertible}
\end{multline}

Thus, we have
\begin{align}
    \tr(B_i) = \left(\sum_{k=1}^{v}Q_{P^{(\alpha)}}(\mathcal{Y}_{\mathrm{P}}^{(k)}) A_{i,k}\right) + Q_{P^{(\alpha)}}(\mathcal{Y}_{\mathrm{I}})A_i'.
\end{align}
The desired result \eqref{eq:TraceFIRestPropToJi} follows from the direct calculation.

\section{Detailed Proof for Section~\ref{sec:achievabilityuBD}}
\subsection{Proof of Proposition~\ref{prop:scoreEstLinear}}\label{app:proofScoreEstLinear}
We show \eqref{eq:scoreEstLinear} separately for each of $i=1,2,3$.

\textbf{1) $i=1$:}
For each $x \in \mathcal{X}_{\mathrm{S}}$ and $k \in [v]$, let
\begin{align}
    \mathcal{Y}_{\mathrm{P}}^{(k,x)}&:=\{y \in \mathcal{Y}_{\mathrm{P}}^{(k)}:x \in y\}, \\
    \mathcal{Y}_{\mathrm{P}}^{(k,\neg x)}&:=\mathcal{Y}_{\mathrm{P}}^{(k)}\backslash\mathcal{Y}_{\mathrm{P}}^{(k,x)}.
\end{align}
We have
\begin{align}
    (\Pi_1(P-P^{(\alpha)}))_x &= \begin{cases}
        P_x-\frac{1}{v}P(\mathcal{X}_{\mathrm{S}}) & (\text{if }x \in \mathcal{X}_{\mathrm{S}})\\
        0 & (\text{if }x \in \mathcal{X}_{\mathrm{N}})
    \end{cases}.
\end{align}
Hence, from Table~\ref{tab:compScoreFuncAndProj}, we deduce that showing \eqref{eq:scoreEstLinear} is equivalent to showing the equality below for all $x \in \mathcal{X}_{\mathrm{S}}$:
\begin{align}
    &\sum_{k=1}^{v} \frac{(e^\epsilon-1)\left((v-k)Q_P(\mathcal{Y}_{\mathrm{P}}^{(k,x)})-kQ_P(\mathcal{Y}_{\mathrm{P}}^{(k,\neg x)})\right)}{\alpha k(e^\epsilon-1)+v}\nonumber \\
    &= \left(P_x-\frac{1}{v}P(\mathcal{X}_{\mathrm{S}})\right)\frac{v-1}{M_1(\alpha,t)}.\label{eq:scoreEstLinearSubspace1}
\end{align}

Using symmetries of block designs and Lemma~\ref{bdlemma}, we can see that for every $x,x' \in \mathcal{X}_{\mathrm{S}}$ and $x'' \in \mathcal{X}_{\mathrm{N}}$ such that $x \neq x'$, we have
\begin{align}
        Q(\mathcal{Y}_{\mathrm{P}}^{(k,x)}|x)&=t_k \frac{r_k e^\epsilon}{r_k(e^\epsilon-1)+b_k}=t_k\frac{ke^\epsilon}{ke^\epsilon+v-k},\\
        Q(\mathcal{Y}_{\mathrm{P}}^{(k,x)}|x')&=t_k\frac{\lambda_k e^\epsilon + r_k-\lambda_k}{r_k(e^\epsilon-1)+b_k}\\
        &=t_k\frac{k((k-1)e^\epsilon+v-k)}{(v-1)(ke^\epsilon+v-k)},\\
        Q(\mathcal{Y}_{\mathrm{P}}^{(k,x)}|x'') &=t_k \frac{r_k}{r_k(e^\epsilon-1)+b_k}=t_k\frac{k}{ke^\epsilon+v-k},
\end{align}
and
\begin{align}
    Q(\mathcal{Y}_{\mathrm{P}}^{(k,\neg x)}|x) &= t_k-Q(\mathcal{Y}_{\mathrm{P}}^{(k,x)}|x),\\
    Q(\mathcal{Y}_{\mathrm{P}}^{(k,\neg x)}|x') &= t_k-Q(\mathcal{Y}_{\mathrm{P}}^{(k,x)}|x'),\\
    Q(\mathcal{Y}_{\mathrm{P}}^{(k,\neg x)}|x'') &= t_k \frac{b_k-r_k}{r_k(e^\epsilon-1)+b_k}\\
    &=t_k\frac{v-k}{ke^\epsilon+v-k}.
\end{align}
By direct calculation, we can see that for every $x \in \mathcal{X}_{\mathrm{S}}$ and $x' \in \mathcal{X}$, \eqref{eq:scoreEstLinearSubspace1} holds for $P=\delta^{(x';v)}$. Since both sides of \eqref{eq:scoreEstLinearSubspace1} are linear in $P$, we conclude that \eqref{eq:scoreEstLinearSubspace1} holds for every $x \in \mathcal{X}_{\mathrm{S}}$ and $P \in \mathcal{P}(\mathcal{X})$.

\textbf{2) $i=2$:}
We have
\begin{align}
    (\Pi_2(P-P^{(\alpha)}))_x &= \begin{cases}
        0 & (\text{if }x \in \mathcal{X}_{\mathrm{S}})\\
        P_x-\frac{P(\mathcal{X}_{\mathrm{N}})}{w-v} & (\text{if }x \in \mathcal{X}_{\mathrm{N}})
    \end{cases}.
\end{align}
Similarly to the $i=1$ case, we deduce that showing \eqref{eq:scoreEstLinear} is equivalent to showing the equality below for all $x \in \mathcal{X}_{\mathrm{N}}$:
\begin{align}
    &\frac{(w-v-1)Q_P(\{x\})-Q_P(\mathcal{Y}_{\mathrm{I}}\backslash\{\{x\}\})}{1-\alpha}\nonumber\\
    &=\left(P_x-\frac{1}{w-v}P(\mathcal{X}_{\mathrm{N}})\right)\frac{w-v-1}{M_2(\alpha,t)}.\label{eq:scoreEstLinearSubspace2}
\end{align}
We can see that for each $x \in \mathcal{X}_{\mathrm{N}}$,
\begin{align}
    Q_P(\{x\}) &= P_x \left(1-\sum_{k=1}^{v}t_k\frac{b_k}{r_k(e^\epsilon-1)+b_k}\right)\\
    &=P_x \sum_{k=1}^{v}t_k\frac{k(e^\epsilon-1)}{ke^\epsilon+v-k}\label{eq:probInvElem}\\
    &=P_x \cdot \frac{(1-\alpha)(w-v-1)}{(w-v)M_2(\alpha,t)}.
\end{align}
This also implies
\begin{multline}
    Q_P(\mathcal{Y}_{\mathrm{I}}\backslash\{\{x\}\})=\\
    (P(\mathcal{X}_{\mathrm{N}})-P_x)\cdot \frac{(1-\alpha)(w-v-1)}{(w-v)M_2(\alpha,t)}.
\end{multline}
Thus, \eqref{eq:scoreEstLinearSubspace2} clearly holds for every $x \in \mathcal{X}_{\mathrm{N}}$.

\textbf{3) $i=3$:}
We have
\begin{align}
    (\Pi_3(P-P^{(\alpha)}))_x &= \begin{cases}
        \frac{P(\mathcal{X}_{\mathrm{S}})-\alpha}{v} & (\text{if }x \in \mathcal{X}_{\mathrm{S}})\\
        -\frac{P(\mathcal{X}_{\mathrm{S}})-\alpha}{w-v} & (\text{if }x \in \mathcal{X}_{\mathrm{N}})
    \end{cases}.
\end{align}
Similarly, we deduce that showing \eqref{eq:scoreEstLinear} is equivalent to showing the equality below:
\begin{align}
    &\left(\sum_{k=1}^{v}\frac{k(e^\epsilon-1)Q_P(\mathcal{Y}_{\mathrm{P}}^{(k)})}{w(\alpha k(e^\epsilon-1)+v)}\right)-\frac{Q_P(\mathcal{Y}_{\mathrm{I}})}{w(1-\alpha)} \nonumber\\
    &=\frac{P(\mathcal{X}_{\mathrm{S}})-\alpha}{v(w-v)M_3(\alpha,t)}. \label{eq:scoreEstLinearSubspace3}
\end{align}
From \eqref{eq:probInvElem}, we have
\begin{align}
    Q_P(\mathcal{Y}_{\mathrm{I}})=(1-P(\mathcal{X}_{\mathrm{S}})) \sum_{k=1}^{v}t_k\frac{k(e^\epsilon-1)}{ke^\epsilon+v-k}.
\end{align}
Also, we have
\begin{align}
    Q_P(\mathcal{Y}_{\mathrm{P}}^{(k)})&=t_k \left(P(\mathcal{X}_{\mathrm{S}})+\frac{b_k}{r_k(e^\epsilon-1)+b_k}P(\mathcal{X}_{\mathrm{N}})\right)\\
    &=t_k \cdot\frac{k(e^\epsilon-1)P(\mathcal{X}_{\mathrm{S}})+v}{ke^\epsilon+v-k}.
\end{align}
Then \eqref{eq:scoreEstLinearSubspace3} follows from the direct calculation.

\subsection{Derivation of (\ref{eq:uBDworstInPAlpha}) in Proposition~\ref{prop:uBDEstError}}\label{appsubsec:uBDworstInPAlpha}
By continuity, it suffices to show \eqref{eq:uBDworstInPAlpha} for $\alpha \in (0,1)$. Hence, we assume $\alpha \in (0,1)$ from now on. 

Let $X\sim P$ and $Y\sim Q(\cdot|X)$. Since $\hat{P}$ is unbiased, we have
\begin{align}
    R_1(Q,\hat{P};P) &= \VAR(\hat{P}_1(Y))\\
    &=\VAR(\hat{P}_1(Y)-P^{(\alpha)})\\
    &=\mathbb{E}\Vert \hat{P}_1(Y)-P^{(\alpha)} \Vert^2-\Vert P-P^{(\alpha)} \Vert^2. \label{eq:uBDMSEDecomp1}
\end{align}
By the definition of the proposed estimator and the law of total expectation, we have
\begin{align}
    &\mathbb{E}\Vert \hat{P}_1(Y)-P^{(\alpha)} \Vert^2 \\
    =& \sum_{i=1}^{3} \frac{(M_i(\alpha,t))^2}{d_i^2} {\mathbb{E}\Vert \Pi_i(\eta_{P^{(\alpha)},Q}(Y))\Vert^2} \label{eq:uBDMSEDecomp2}\\
    =& \sum_{i=1}^{3} {\sum_{x \in \mathcal{X}} \frac{(M_i(\alpha,t))^2}{d_i^2} P_x \mathbb{E}[\Vert \Pi_i(\eta_{P^{(\alpha)},Q}(Y))\Vert^2|X=x]}.
\end{align}
As we presented in \eqref{eq:normScoreProtected} and \eqref{eq:normScoreInvertible}, $\Vert \Pi_i(\eta_{P^{(\alpha)},Q}(y))\Vert^2$ depends on $y$ only by (i) whether $y \in \mathcal{Y}_{\mathrm{P}}$ or $\mathcal{Y}_{\mathrm{I}}$, and (ii) the value of $|y|$. 
By the regularity of block designs, the conditional distributions of $\Vert \Pi_i(\eta_{P^{(\alpha)},Q}(Y))\Vert^2$ given $X=x$ are the same for all $x \in \mathcal{X}_{\mathrm{S}}$. Also, by the definition of extremal ULDP mechanisms, the conditional distributions of $\Vert \Pi_i(\eta_{P^{(\alpha)},Q}(Y))\Vert^2$ given $X=x$ are the same for all $x \in \mathcal{X}_{\mathrm{N}}$. 
Thus, we can write
\begin{align}
    \mathbb{E}\Vert \hat{P}_1(Y)-P^{(\alpha)} \Vert^2 &= \nu_1 P(\mathcal{X}_{\mathrm{S}})+\nu_2 P(\mathcal{X}_{\mathrm{N}})\\
    &=\nu_1 \beta + \nu_2 (1-\beta),
\end{align}
where $\nu_1,\nu_2 \geq 0$ are given by
\begin{align}
    \nu_1 = \sum_{i=1}^{3} \frac{(M_i(\alpha,t))^2}{d_i^2}{\mathbb{E}[\Vert \Pi_i(\eta_{P^{(\alpha)},Q}(Y))\Vert^2|X=x]},\\
    \nu_2 = \sum_{i=1}^{3} \frac{(M_i(\alpha,t))^2}{d_i^2}{\mathbb{E}[\Vert \Pi_i(\eta_{P^{(\alpha)},Q}(Y))\Vert^2|X=x']},
\end{align}
for any of $x \in \mathcal{X}_{\mathrm{S}}$ and $x' \in \mathcal{X}_{\mathrm{N}}$.
Thus, $\mathbb{E}\Vert \hat{P}_1(Y)-P^{(\alpha)} \Vert^2$ is a constant among all $P \in \mathcal{P}(\mathcal{X})$ such that $P(\mathcal{X}_\mathrm{S})=\beta$.

Moreover, we have
\begin{align}
    &\Vert P-P^{(\alpha)} \Vert^2 \\
    = &\sum_{x \in \mathcal{X}_\mathrm{S}}\left(P_x-\frac{\alpha}{v}\right)^2 + \sum_{x \in \mathcal{X}_\mathrm{N}}\left(P_x-\frac{1-\alpha}{w-v}\right)^2\\
    = & \frac{\alpha^2}{v^2}+\frac{(1-\alpha)^2}{(w-v)^2}-\frac{2\alpha}{v}\beta-\frac{2(1-\alpha)}{w-v}(1-\beta)\nonumber\\
    &+\sum_{x \in \mathcal{X}_\mathrm{S}}P_x^2+\sum_{x \in \mathcal{X}_{\mathrm{N}}}P_x^2.
\end{align}
By the Cauchy-Schwarz inequality, we have
\begin{align}
    \sum_{x \in \mathcal{X}_\mathrm{S}}P_x^2 \geq \frac{\beta^2}{v^2}, \quad \sum_{x \in \mathcal{X}_\mathrm{N}}P_x^2 \geq \frac{(1-\beta)^2}{(w-v)^2},
\end{align}
and the equality holds if $P=P^{(\beta)}$. 
From this, we conclude that among all $P \in \mathcal{P}(\mathcal{X})$ such that $P(\mathcal{X}_\mathrm{S})=\beta$, $R_1(Q,\hat{P};P)$ is maximized at $P=P^{(\beta)}$, which shows \eqref{eq:uBDworstInPAlpha}.

\subsection{Derivation of (\ref{eq:uBDMSEAtPAlpha}) in Proposition~\ref{prop:uBDEstError}}\label{appsubsec:uBDMSEAtPAlpha}
Note that except for the case of $v \geq 2$ and $t = \delta^{(v;v)}$, $M(\alpha,t)$ is continuously differentiable with respect to $\alpha$ in a neighborhood of $[0,1]$. Thus, again by continuity, it suffices to show \eqref{eq:uBDMSEAtPAlpha} for $\alpha \in (0,1)$. Hence, we assume $\alpha \in (0,1)$ from now on. 

From \eqref{eq:uBDMSEDecomp1}, we have
\begin{align}
    &R_1(Q,\hat{P};P^{(\beta)})\\
    &=\mathbb{E}_{Y \sim Q_{P^{(\beta)}}}\Vert\hat{P}_1(Y)-P^{(\alpha)}\Vert^2 - \Vert P^{(\beta)}-P^{(\alpha)}\Vert^2.
\end{align}
It is straightforward to see that
\begin{align}
    \Vert P^{(\beta)}-P^{(\alpha)}\Vert^2 = \frac{w}{v(w-v)}(\beta-\alpha)^2.
\end{align}
Hence, it remains to show that
\begin{multline}
    \mathbb{E}_{Y \sim Q_{P^{(\beta)}}}\Vert\hat{P}_1(Y)-P^{(\alpha)}\Vert^2\\
    = (\beta-\alpha)F(\alpha,t)+M(\alpha,t).
\end{multline}
Since both sides of the above equation are affine in $\beta$, it suffices to show that
\begin{align}
    \mathbb{E}_{Y \sim Q_{P^{(\alpha)}}}\Vert\hat{P}_1(Y)-P^{(\alpha)}\Vert^2 &= M(\alpha,t),\label{eq:errorBetaConstTermM}\\
    \frac{\partial}{\partial \beta}  \mathbb{E}_{Y \sim Q_{P^{(\beta)}}}\Vert\hat{P}_1(Y)-P^{(\alpha)}\Vert^2 &= F(\alpha,t).\label{eq:ErrorBetaLinearF}
\end{align}
\paragraph{Proof of \eqref{eq:errorBetaConstTermM}}
It is easy to see that $t(Q)$ in \eqref{eq:TraceFIRestPropToJi} is equal to $t$. Hence, by \eqref{eq:uBDMSEDecomp2}, \eqref{eq:trFIandNormProj}, and \eqref{eq:TraceFIRestPropToJi}, we have
\begin{align}
    &\mathbb{E}_{Y \sim Q_{P^{(\alpha)}}}\Vert\hat{P}_1(Y)-P^{(\alpha)}\Vert^2 \\
    &= \sum_{i=1}^{3} \frac{(M_i(\alpha,t))^2}{d_i^2} {\mathbb{E}_{Y \sim Q_{P^{(\alpha)}}}\Vert \Pi_i(\eta_{P^{(\alpha)},Q}(Y))\Vert^2}\\
    &= \sum_{i=1}^{3}M_i(\alpha,t)=M(\alpha,t),
\end{align}
which shows \eqref{eq:errorBetaConstTermM}.

\paragraph{Proof of \eqref{eq:ErrorBetaLinearF}}
Each $Q_{P^{(\beta)}}(y)$ is affine in $\beta$, and hence we can write
\begin{align}
    Q_{P^{(\beta)}}(y) = f_1(y)\beta + f_2(y)
\end{align}
for $f_1,f_2:\supp(Q) \rightarrow \mathbb{R}$. Then,
\begin{multline}
    \frac{\partial}{\partial \beta}  \mathbb{E}_{Y \sim Q_{P^{(\beta)}}}\Vert\hat{P}_1(Y)-P^{(\alpha)}\Vert^2 \\
    = \sum_{y} f_1(y) \Vert \hat{P}_1(y)-\hat{P}(\alpha)\Vert^2. \label{eq:errorLinearCalc}
\end{multline}
Also, from \eqref{eq:errorBetaConstTermM}, we have
\begin{align}
    F(\alpha,t)=&\sum_{y} \frac{\partial}{\partial \alpha}\left(Q_{P^{(\alpha)}}(y) \Vert \hat{P}_1(y)-P^{(\alpha)} \Vert^2\right)\\
    =&\sum_{y} f_1(y)\Vert \hat{P}_1(y)-{P}^{(\alpha)}\Vert^2 \nonumber \\
    &+\sum_{y} Q_{P^{(\alpha)}}(y) \left\langle \hat{P}_1(y)-P^{(\alpha)}, \frac{\partial P^{(\alpha)}}{\partial \alpha} \right\rangle.
\end{align}
Since $\hat{P}$ is unbiased, $\sum_{y} Q_{P^{(\alpha)}}(y) \hat{P}_1(y)=P^{(\alpha)}$. Thus, combining with \eqref{eq:errorLinearCalc}, we obtain \eqref{eq:ErrorBetaLinearF}. 

This ends the derivation of \eqref{eq:uBDMSEAtPAlpha}.

\section{Proof of Theorem~\ref{thm:closedFormExpPreciseForm}}\label{app:proofClosedFormExpPreciseForm}
Let $(\alpha^*,t^*)$ be as specified in Theorem~\ref{thm:closedFormExpPreciseForm}, with $t^*=\delta^{(k^*;v)}$. 
Here, $k^*=1$ for Case (a), and $k^* \in K^*(v,\epsilon) \cap [2:v-1]$ for Case (b).
We separately show two equalities,
\begin{align}
    M(\alpha^*,t^*)=\inf\limits_{t \in \Delta_v} M(\alpha^*,t), \label{eq:ClosedFormOptOverT}\\
    M(\alpha^*,t^*)=\sup\limits_{\alpha \in [0,1]} M(\alpha,t^*).\label{eq:ClosedFormOptOverAlpha}
\end{align}

\subsection{Proof of (\ref{eq:ClosedFormOptOverT})}
\textbf{Case (a)-Condition (i):}
This case is trivial, since when $v=1$, $\Delta_v$ is a singleton.

\textbf{Case (a)-Condition (ii):}
In this case, we have
\begin{align}
    e^\epsilon \geq w-v+\sqrt{\frac{(w-1)(w-2)}{2}} &\geq 1+w-v,
\end{align}
and thus $\alpha^*=\frac{v(e^\epsilon-1-w+v)}{w(e^\epsilon-1)}$.

First, observe that by a calculation, we have
\begin{align}
    M_i(\alpha^*,t^*)=d_i \frac{(e^\epsilon+v-1)^2}{w(e^\epsilon-1)^2}, \quad \forall i=1,2,3.
\end{align}
Recall that $(d_1,d_2,d_3)=(v-1,w-v-1,1)$.
Motivated from this, we apply the arithmetic mean-harmonic mean inequality, obtaining the following.
\begin{align}
    &M(\alpha^*, t) \\
    =& (d_1+d_2+d_3)\sum_{i=1}^{3}\frac{d_i}{d_1+d_2+d_3} (M_i(\alpha^*,t)/d_i)\\
    \geq& \frac{(d_1+d_2+d_3)^2}{\sum_{i=1}^{3}d_i^2/ M_i(\alpha^*, t)}.
\end{align}
This becomes an equality when $t=t^*=\delta^{(1;v)}$. Also, we can see that $\sum_{i=1}^{3}d_i^2 /M_i(\alpha^*, t) = \sum_{k=1}^{v} g_k t_k$, where $g_k \geq 0$ are given by
\begin{align}
    g_k = &\frac{kv(v-k)(e^\epsilon-1)^2}{(\alpha^* k(e^\epsilon-1)+v)(ke^\epsilon+v-k)}\nonumber\\
    &+\frac{k(w-v)(w-v-1)(e^\epsilon-1)}{(1-\alpha^*)(ke^\epsilon+v-k)}\nonumber\\
    &+\frac{vk(w-v)(e^\epsilon-1)}{w(1-\alpha^*)(\alpha^* k(e^\epsilon-1)+v)}.
\end{align}
Now, we claim that
\begin{align}
    g_1 \geq g_k, \quad \forall k\in[2:v], \label{eq:uRROptCoeff1Max}
\end{align}
because this claim deduces \eqref{eq:ClosedFormOptOverT} as 
\begin{align}
    \inf_{t \in \Delta_v}M(\alpha^*, t)=M(\alpha^*, t^*)=\frac{(d_1+d_2+d_3)^2}{g_1}.
\end{align}
To show \eqref{eq:uRROptCoeff1Max}, we calculate $g_1-g_k$, and the result is
\begin{align}
    &g_1-g_k = \\
    &\frac{vk(k-1)w(e^\epsilon-1)^2}{\zeta^2 (ke^\epsilon+v-k)(k(e^\epsilon-1-w+v)+w)}\nonumber\\
    &\times \left(\zeta^2-2(w-1)\zeta+w(w-1)\frac{k-1}{k} \right),
\end{align}
where $\zeta:=e^\epsilon+v-1$. Note that $\frac{k-1}{k}$ increases in $k$. Hence, to show $g_1-g_k \geq 0$ for all $k \in [2:v]$, it suffices to show that
\begin{align}
    \zeta^2-2(w-1)\zeta+\frac{w(w-1)}{2} \geq 0. \label{eq:uRROptQuadPoly}
\end{align}
The LHS of the above equation is a quadratic polynomial in $\zeta$ with roots $w-1\pm\sqrt{\frac{(w-1)(w-2)}{2}}$. Hence, \eqref{eq:uRROptQuadPoly} holds for $\zeta \geq w-1+\sqrt{\frac{(w-1)(w-2)}{2}}$, that is, $e^\epsilon \geq w-v+\sqrt{\frac{(w-1)(w-2)}{2}}$. This shows \eqref{eq:uRROptCoeff1Max}, and hence ends the proof.

\textbf{Case (a)-Condition (iii):}
In this case, we have
\begin{align}
    e^\epsilon \leq 1+\sqrt{\frac{2(w-2)}{w-1}} \leq 1+w-v=w-1,
\end{align}
and thus $\alpha^*=0$. By a direct calculation, with $t_2=1-t_1$, we obtain
\begin{align}
    M(\alpha^*, t) &= f(t_1)\\
    :=\frac{e^\epsilon+1}{(e^\epsilon-1)^2}\frac{1}{t_1} &+ \frac{e^\epsilon(e^\epsilon+1)}{(w-2)(e^\epsilon-1)}\frac{w-3}{e^\epsilon+1-t_1}\nonumber\\
    &+ \frac{w}{(w-2)(e^\epsilon-1)}\frac{1}{2-t_1}.
\end{align}
Since $f(t_1)$ is convex in $t_1$, to show the optimality of $t^*=\delta^{(1;v)}$, it suffices to show that $f'(1) \leq 0$. A straightforward calculation shows that
\begin{align}
    f'(1) = \frac{w-1}{e^\epsilon (e^\epsilon-1)^2(w-2)}\left(e^{2\epsilon}-2e^\epsilon-\frac{w-3}{w-1} \right).
\end{align}
Here, $e^{2\epsilon}-2e^\epsilon-\frac{w-3}{w-1}$ is a quadratic polynomial in $e^\epsilon$ with roots $1 \pm \sqrt{\frac{2(w-2)}{w-1}}$. Thus, $f'(1) \leq 0$ for $e^\epsilon \leq 1+\sqrt{\frac{2(w-2)}{w-1}}$, which ends the proof.

\textbf{Case (b):}
Since $M_2(1,t)=M_3(1,t)=0$, we have
\begin{align}
    \inf_{t \in \Delta_v} M(\alpha^*,t) &= \inf_{t \in \Delta_v} M_1(1,t) \label{eq:ULDPtoLDPAlphaOne1}\\
    &= \inf_{t \in \Delta_v }\frac{(v-1)^2}{v(e^\epsilon-1)^2}\frac{1}{\sum_{k=1}^v t_k \frac{k(v-k)}{(ke^\epsilon+v-k)^2}}\label{eq:ULDPtoLDPAlphaOne2}\\
    &= \min_{k \in [v-1]} \frac{(v-1)^2}{v(e^\epsilon-1)^2} \frac{(ke^\epsilon+v-k)^2}{k(v-k)}\label{eq:ULDPtoLDPAlphaOne3}\\
    &= \min_{k \in [v-1]} R^\mathrm{BD}(v,k,\epsilon), \label{eq:ULDPtoLDPAlphaOne4}
\end{align}
where $R^\mathrm{BD}$ is in \eqref{eq:defRBD}, and optimal $k$ in \eqref{eq:ULDPtoLDPAlphaOne4} and optimal $t$ in \eqref{eq:ULDPtoLDPAlphaOne1} are related by $t=\delta^{(k;v)}$. Thus, by the definition of $K^*$, $t^* = \delta^{(k^* ;v)}$ is an optimal solution to \eqref{eq:ULDPtoLDPAlphaOne1}. 

\subsection{Proof of (\ref{eq:ClosedFormOptOverAlpha})}
Let $(Q,\hat{P})$ be a $(w,v,\epsilon,k^*)$-simple uBD scheme.
By Proposition~\ref{prop:simpleuBDEstInvarOverAlpha}, Equation~\eqref{eq:uBDMSEAtPAlpha} with $t=t^*=\delta^{(k^*;v)}$ holds for all $\alpha \in [0,1]$. Setting $\alpha=\beta$ in \eqref{eq:uBDMSEAtPAlpha} gives
\begin{align}
    R_1(Q,\hat{P};P^{(\beta)})=M(\beta,\delta^{(k^*;v)}), \quad \forall \beta \in [0,1].
\end{align}
Thus, it suffices to show that $\sup_{\beta \in [0,1]} R_1(Q,\hat{P};P^{(\beta)})$ is attained at $\beta=\alpha^*$.
For Case (a), it directly follows from the worst-case analysis of uRR \cite[Propositions 16, 17]{murakami2019utility}.
For Case (b), we claim that
\begin{align}
    F(1,\delta^{(k;v)}) \geq 0, \quad \forall k \in [2:v-1].\label{eq:simpleuBDk2FPositive}
\end{align}
Then, by setting $\alpha=1$ in \eqref{eq:uBDMSEAtPAlpha}, we obtain the desired result. 
To show \eqref{eq:simpleuBDk2FPositive}, we calculate $F(1,\delta^{(k;v)})$, and the result is
\begin{align}
    F(1,\delta^{(k;v)})=
    \frac{ke^\epsilon+v-k}{v(e^\epsilon-1)}\left(\frac{(v-1)^2}{v-k}-\frac{v+1}{k}\right).
\end{align}
Hence, it suffices to show that
\begin{align}
    (v-1)^2 k - (v+1)(v-k) \geq 0, \quad \forall k \in [2:v-1].
\end{align}
Since the LHS of the above inequality is increasing in $k$, it suffices to show only for $k=2$, that is,
\begin{align}
    2(v-1)^2-(v+1)(v-2) \geq 0.
\end{align}
We have
\begin{equation}
   2(v-1)^2-(v+1)(v-2) = v^2-3v+4.
\end{equation}
As a quadratic polynomial in $v$, above equation has the discriminant $3^2-4\cdot4<0$, and therefore it is non-negative for every $v$. Thus, we show \eqref{eq:simpleuBDk2FPositive}, and hence we end the proof of \eqref{eq:ClosedFormOptOverAlpha}.

\section{Proof of Corollary~\ref{cor:uSSStrictSuboptimal}}\label{app:compBwuSS}
Referring \cite[Equation (35)]{he2025addressing}, the estimation error of uSS is
\begin{align}
    &n \cdot R_n(Q_{\mathrm{uSS},k}, \hat{P}_{\mathrm{uSS},k};P) \\
    =& L_1(k)+P(\mathcal{X}_{\mathrm{S}})L_2(k)\nonumber\\
    &+P(\mathcal{X}_{\mathrm{N}})L_3(k)+1-\sum_{x \in \mathcal{X}}P_x^2,\label{eq:uSSMSE}
\end{align}
where
\begin{align}
    L_1(k)&=\frac{v(ke^\epsilon-e^\epsilon+v-k)(ke^\epsilon-k+v-1)}{k(v-k)(e^\epsilon-1)^2},\\
    L_2(k)&=\frac{k(1-k)(e^\epsilon-1)+(v-1)(v-2k)}{k(v-k)(e^\epsilon-1)},\\
    L_3(k)&=\frac{v}{k(e^\epsilon-1)}.
\end{align}
Especially, letting $P$ be the uniform distribution on $\mathcal{X}_{\mathrm{S}}$, we have
\begin{align}
    &n \cdot R_n(Q_{\mathrm{uSS},k}, \hat{P}_{\mathrm{uSS},k}) \\
    \geq &n \cdot  R_n(Q_{\mathrm{uSS},k}, \hat{P}_{\mathrm{uSS},k};P)\\
    =& L_1(k)+L_2(k)+1-\frac{1}{v},
\end{align}
and hence
\begin{align}
    R(Q_{\mathrm{uSS},k}, \hat{P}_{\mathrm{uSS},k}) \geq  L_1(k)+L_2(k)+1-\frac{1}{v}.
\end{align}
Recall from Remark~\ref{rmk:ULDPtoLDPTight} that $M^*(w,v,\epsilon)=M^*(v,v,\epsilon)=\min_{k \in [v-1]} R^{\mathrm{BD}}(v,k,\epsilon)$. 
By a straightforward calculation, we obtain
\begin{align}
    &\left(L_1(k)+L_2(k)+1-\frac{1}{v}\right) - R^{\mathrm{BD}}(v,k,\epsilon) \\
    &= \frac{2(k-1)}{v-k}.
\end{align}
Hence, we have
\begin{align}
    R(Q_{\mathrm{uSS},k}, \hat{P}_{\mathrm{uSS},k}) \geq R^{\mathrm{BD}}(v,k,\epsilon)+\frac{2(k-1)}{v-k}.
\end{align}
Using this, we argue that $R(Q_{\mathrm{uSS},k}, \hat{P}_{\mathrm{uSS},k})>M^*(w,v,\epsilon)$ for every $k \in [v-1]$, as follows.
\begin{enumerate}
    \item If $k=1$, then since $\epsilon<\ln\sqrt{\frac{(v-1)(v-2)}{2}}=E(v,1)$, we have $1 \notin K^*(v,\epsilon)$, and hence $R^{\mathrm{BD}}(v,k,\epsilon)>M^*(w,v,\epsilon)$. Thus, $R(Q_{\mathrm{uSS},k}, \hat{P}_{\mathrm{uSS},k})>M^*(w,v,\epsilon)$.
    \item If $k \geq 2$, then $R^{\mathrm{BD}}(v,k,\epsilon) \geq M^*(w,v,\epsilon)$ and $\frac{2(k-1)}{v-k}>0$. Thus, $R(Q_{\mathrm{uSS},k}, \hat{P}_{\mathrm{uSS},k})>M^*(w,v,\epsilon)$.
\end{enumerate}
This ends the proof of Corollary~\ref{cor:uSSStrictSuboptimal}.

\begin{remark}\label{rmk:freq2DistEst}
    Originally, \cite{he2025addressing} formulated the estimator error in terms of the MSE between the empirical frequency of $X^n$ and the estimated distribution. That is, letting $T^{(n)}:[w]^n \rightarrow \Delta_w$, $T^{(n)}_{x'}(x^n)=\frac{1}{n}\sum_{i=1}^{n} \mathbbm{1}(x_i=x')$, they formulated
\begin{multline}
    R_n'(Q,\hat{P};x^n):=\\
    \mathbb{E}_{Y^n \sim \prod_{i=1}^{n}Q(\cdot|x_i)}\Vert \hat{P}_n(Y^n)-T^{(n)}(x^n)\Vert^2.
\end{multline}
However, provided that the scheme is unbiased, the MSE for the frequency estimation can be translated into the MSE for the distribution estimation, by using the law of total variance. 
Specifically, for given $P \in \Delta_w$ and an \emph{unbiased} scheme $(Q,\hat{P})$, let $X^n \sim P^n$ and $Y^n\sim \prod_{i=1}^{n}Q(\cdot|X_i)$ given $X^n$. Then
\begin{align}
    &R_n(Q,\hat{P};P)=\VAR\left[\hat{P}_n(Y^n)\right]\\
    =&\mathbb{E}[\VAR[\hat{P}_n(Y^n)|X^n]]+\VAR[\mathbb{E}[\hat{P}_n(Y^n)|X^n]]\\
    =&\mathbb{E}[R_n'(Q,\hat{P};X^n)]+\VAR[T^{(n)}(X^n)].
\end{align}
Here, each of $T^{(n)}_{x}(X^n)$ is a binomial random variable with $n$ trials and probability of success $P_{x}$, divided by $n$. Hence
\begin{align}
    &\VAR[T^{(n)}(X^n)] \\
    = &\frac{1}{n}\sum_{x \in \mathcal{X}}(P_x-P_x^2)=\frac{1}{n}\left(1-\sum_{x \in \mathcal{X}} P_x^2\right).
\end{align}
In conclusion, we have
\begin{align}
    &R_n(Q,\hat{P};P) \\
    = &\mathbb{E}_{X^n \sim P^n}[R_n'(Q,\hat{P};X^n)] + \frac{1}{n}\left(1-\sum_{x \in \mathcal{X}} P_x^2\right).\label{eq:freq2DistTrans}
\end{align}
From this, we deduce \eqref{eq:uSSMSE} as follows.
In \cite[Equation (35)]{he2025addressing}, it is provided that
\begin{align}
    &n \cdot R_n'(Q_{\mathrm{uSS},k}, \hat{P}_{\mathrm{uSS},k};x^n)\\
    =&L_1(k)
    +\sum_{x' \in \mathcal{X}_{\mathrm{S}}}T^{(n)}_{x'}(x^n)  L_2(k)\nonumber\\
    &+\sum_{x' \in \mathcal{X}_{\mathrm{N}}}T^{(n)}_{x'}(x^n)L_3(k).
\end{align}
Under $X^n \sim P^n$, it is straightforward that $\mathbb{E}[T^{(n)}_{x'}(X^n)]=P_{x'}$ for every $x' \in \mathcal{X}$. Hence
\begin{align}
    &n \cdot \mathbb{E}_{X^n \sim P^n}[R_n'(Q_{\mathrm{uSS},k}, \hat{P}_{\mathrm{uSS},k};X^n)]\\
    =&L_1(k)+P(\mathcal{X}_{\mathrm{S}})L_2(k)+P(\mathcal{X}_{\mathrm{N}})L_3(k).
\end{align}
Applying \eqref{eq:freq2DistTrans}, we obtain \eqref{eq:uSSMSE}.

\end{remark}

\end{document}